\documentclass[10pt,twocolumn]{IEEEtran}

\usepackage{amsfonts}
\usepackage[latin2]{inputenc}
\usepackage{t1enc}
\usepackage{amssymb}
\usepackage{amsmath}
\usepackage[mathscr]{eucal}
\usepackage{amsthm}
\usepackage{array}
\usepackage{color}
\usepackage{cite}

\newtheorem{definition}{Definition}[section]
\newtheorem{thm}[definition]{Theorem}
\newtheorem{lemma}[definition]{Lemma}
\newtheorem{rem}[definition]{Remark}
\newtheorem{cor}[definition]{Corollary}
\newtheorem{ex}[definition]{Example}

\newcommand{\be}{\begin{equation}}
\newcommand{\ee}{\end{equation}}
\newcommand{\bea}{\begin{eqnarray}}
\newcommand{\eea}{\end{eqnarray}}
\newcommand{\beas}{\begin{eqnarray*}}
\newcommand{\eeas}{\end{eqnarray*}}

\def\vfi{\varphi}
\def\hil{{\mathcal H}}
\def\kil{{\mathcal K}}
\def\A{{\mathcal A}}
\def\B{{\mathcal B}}
\def\C{{\mathcal C}}
\def\E{{\mathcal E}}
\def\F{{\mathcal F}}
\def\I{\mathcal{I}}

\def\S{{\mathcal S}}
\def\X{{\mathcal X}}
\def\Y{{\mathcal Y}}

\def\half{\frac{1}{2}}
\def\iff{\Longleftrightarrow}
\def\imp{\Longrightarrow}
\def\ep{\varepsilon}
\def\bN{\mathbb{N}}
\def\bC{\mathbb{C}}
\def\bR{\mathbb{R}}
\def\bZ{\mathbb{Z}}
\def\Z{\mathbb{Z}}
\def\bz{\left(}
\def\jz{\right)}
\def\inv{^{-1}}
\def\kii{\emph}
\def\kiii{}
\def\unit{I}
\def\egy{\mathbf{1}}
\def\nw{^{*}}
\def\old{^{}}
\def\x{^{(v)}}
\def\xx{(v)}
\def\oldd{\{\s\}}
\def\neww{*}
\def\what{\widehat}
\def\ol{\overline}
\def\nn{\nonumber}
\def\map{\Phi}
\def\dimen{\nu}
\def\dd{\,d}
\def\car{\mathrm{CAR}}
\def\lm{\Lambda}
\def\alm{\ol\lm}
\def\lf{_{\circ}}
\def\sci{\underline{sc}}
\def\scs{\overline{sc}}
\def\t{^{(t)}}

\newcommand{\ki}{\emph}

\newcommand{\s}{\mbox{ }}
\newcommand{\ds}{\mbox{ }\mbox{ }}
\newcommand{\norm}[1]{\left\| #1\right\|}
\newcommand{\inner}[2]{\langle #1 , #2\rangle}
\newcommand{\abs}[1]{\left| #1 \right|}

\newcommand{\vect}[1]{\mathbf{#1}}
\newcommand{\vecc}[1]{\underline{#1}}
\newcommand{\diad}[2]{|#1\rangle\langle #2|}
\newcommand{\pr}[1]{\diad{#1}{#1}}

\newcommand{\rsr}[3]{D_{#3}\bz #1\,\|\, #2\jz}

\newcommand{\rsrn}[3]{D_{#3}\nw\bz #1\,\|\, #2\jz}
\newcommand{\rsrx}[3]{D_{#3}\x\bz #1\,\|\, #2\jz}

\newcommand{\pf}[3]{\psi(#1|#2\|#3)}

\newcommand{\pfo}[3]{\psi\old(#1|#2\|#3)}
\newcommand{\pfn}[3]{\psi\nw(#1|#2\|#3)}
\newcommand{\pfx}[3]{\psi\x(#1|#2\|#3)}

\DeclareMathOperator{\Tr}{Tr}
\DeclareMathOperator{\supp}{supp}
\DeclareMathOperator{\ran}{ran}
\DeclareMathOperator{\Exp}{\mathbb{E}}
\DeclareMathOperator{\Prob}{\mathbb{P}}

\DeclareMathOperator{\derright}{\partial^{+}}
\DeclareMathOperator{\spann}{span}

\renewcommand\theenumi{(\arabic{enumi})}
\renewcommand\theenumii{$\alph{enumii})$}

\begin{document}

\renewcommand\theenumi{(\roman{enumi})}
\renewcommand\theenumii{$\alph{enumii})$}

%\setcounter{page}{1}
%\title{Coding theorems for compound problems via quantum R\'enyi divergences}
%\author{Mil\'an Mosonyi \thanks{M.~Mosonyi is with the F\'{\i}sica Te\`{o}rica: Informaci\'{o} i Fenomens Qu\`{a}ntics,
%Universitat Aut\`{o}noma de Barcelona, ES-08193 Bellaterra (Barcelona), Spain, and with the Mathematical Institute, Budapest University of Technology and Economics,
%Egry J\'ozsef u~1., Budapest, 1111 Hungary (e-mail: milan.mosonyi@gmail.com)}\\
%\thanks{This paper was presented in part at the Conference on the Theory of Quantum Computation, Communication \& Complexity (TQC 2014), 21--23 May 2014, Singapore.}}

\title{Two approaches to obtain the strong converse exponent of quantum hypothesis testing for general sequences of quantum states}
\author{Mil\'an Mosonyi, Tomohiro Ogawa
\thanks{Mil\'an Mosonyi was with the F\'{\i}sica Te\`{o}rica: Informaci\'{o} i Fenomens Qu\`{a}ntics,  Departament de F\'isica,
Universitat Aut\`{o}noma de Barcelona.
%, ES-08193 Bellaterra (Barcelona), Spain. 
He is now with the Technische Universit\"at M\"unchen -- Institute for Advanced Study \& Zentrum Mathematik,
% Lichtenbergstra\"se 2a,
%85748 Garching, Germany, 
and also with the Mathematical Institute, Budapest University of Technology and Economics.
%Egry J\'ozsef u~1., Budapest, 1111 Hungary 
(e-mail: milan.mosonyi@gmail.com)}
\thanks{Tomohiro Ogawa is with the Graduate School of Information Systems,
University of Electro-Communications,
1-5-1 Chofugaoka, Chofu-shi, Tokyo, 182-8585, Japan}

\thanks{Copyright (c) 2014 IEEE. Personal use of this material is permitted.  However, permission to use this material for any other purposes must be obtained from the IEEE by sending a request to pubs-permissions@ieee.org.}
}

%
%\author{Mil\'an Mosonyi}
%\email{milan.mosonyi@gmail.com}
%\affiliation{
%F\'{\i}sica Te\`{o}rica: Informaci\'{o} i Fenomens Qu\`{a}ntics,
%Universitat Aut\`{o}noma de Barcelona, ES-08193 Bellaterra (Barcelona), Spain.
%}
%\affiliation{
%Mathematical Institute, Budapest University of Technology and Economics, \\
%Egry J\'ozsef u~1., Budapest, 1111 Hungary.
%}
%
%\author{Tomohiro Ogawa}
%\email{ogawa@is.uec.ac.jp}
%\affiliation{
%Graduate School of Information Systems,
%University of Electro-Communications,
%1-5-1 Chofugaoka, Chofu-shi, Tokyo, 182-8585, Japan.
%}

\maketitle

\begin{abstract}

We present two general approaches to obtain the strong converse exponent of simple quantum hypothesis testing
for correlated quantum states. One approach requires that the states satisfy a certain factorization property; typical examples
of such states are the temperature states of translation-invariant finite-range interactions on a spin chain.
The other approach requires the differentiability of a regularized R\'enyi $\alpha$-divergence in the parameter $\alpha$; typical examples of such states include temperature states of non-interacting fermionic lattice systems, and classical irreducible Markov chains. In all cases, we get that the strong converse exponent is equal to the Hoeffding anti-divergence, which in turn is obtained from the regularized R\'enyi divergences of the two states.
\end{abstract}

\section{Introduction}

Assume that we have a quantum system with finite-dimensional Hilbert space $\hil$, and we know that the system is either prepared in the state $\rho_1$ (\ki{null-hypothesis} $H_0$), or  in the state $\sigma_1$ (\ki{alternative hypothesis} $H_1$).
We further assume that we have access to several identical copies of the system, either all prepared in state $\rho_1$, or all prepared in state $\sigma_1$; for $n$ copies this means that the state of the system is given by $\rho_n:=\rho_1^{\otimes n}$ or by $\sigma_n:=\sigma_1^{\otimes n}$.
Our task is to decide which hypothesis is true, by performing measurements on the system. It is easy to see that the most general decision scheme can be described by a binary POVM (\ki{positive operator-valued measure}), with POVM elements $T(0)=T$ corresponding to accepting the null-hypothesis, and $T(1)=I-T$ corresponding to
accepting the alternative hypothesis.
Here we assume that we are allowed to make collective measurements on all the available copies, i.e., $T$ can be any positive semidefinite operator on $\hil_n=\hil^{\otimes n}$, satisfying $T\le I$. Such an operator is called a \ki{test}. Obviously, the test $T$ and the POVM $(T,I-T)$ uniquely determine each other.

There are two possible ways of making an erroneous decision; either by accepting $H_1$ when $H_0$ is true (\ki{type I error}), or the other way around (\ki{type II error}).
For a test $T$, the probabilities of these errors are given by
\begin{align*}
\alpha_n(T)&:=\Tr\rho_n(I-T),\ds\ds\text{(type I)\ds\ds\ds\ds and}\\
\beta_n(T)&:=\Tr\sigma_nT,\ds\ds\text{(type II)}.
\end{align*}
Obviously, there is a trade-off between these two error probabilities, and there are various ways to jointly optimize them. Probably the most studied scenario is where the type I error is required to vanish in the asymptotics or, in a different formulation, to stay below a given threshold for all number of copies. The quantum \ki{Stein's lemma} \cite{HP,ON} states that in both formulations, the best achievable asymptotics for the type II error is an exponential decay, where the exponent is given by the relative entropy $D(\rho_1\|\sigma_1)$.

To get a more detailed view of the trade-off between the two error probabilities, one may ask about the asymptotics of the type I error when the type II error is made to
decay as $\sim e^{-nr}$, where $r$ is a fixed rate below or above the optimal rate $D(\rho_1\|\sigma_1)$. As it turns out, for rates below $D(\rho_1\|\sigma_1)$, the best
achievable asymptotics for the type I error is an exponential decay, with a rate
\begin{align}
d(r|\rho\|\sigma)&:=\sup\left\{-\limsup_{n\to+\infty}\frac{1}{n}\log\alpha_n(T_n):\right.\nonumber\\
&\ds\ds\ds\ds\ds\ds\left.\,\limsup_{n\to+\infty}\frac{1}{n}\log\beta_n(T_n)<-r\right\}\label{thm:iid direct rate1}\\
&=H_r(\rho_1\|\sigma_1):=
\sup_{0<\alpha<1}\frac{\alpha-1}{\alpha}\left[r-D_{\alpha}\old(\rho_1\|\sigma_1)\right].\label{thm:iid direct rate}
\end{align}
Here, $d(r|\rho\|\sigma)$ is the \ki{direct exponent} of the problem, $H_r(\rho_1\|\sigma_1)$ is the \ki{Hoeffding divergence}, and
$D_{\alpha}\old(\rho_1\|\sigma_1):=\frac{1}{\alpha-1}\log\Tr\rho_1^{\alpha}\sigma_1^{1-\alpha}$ is a quantum version of R\'enyi's $\alpha$-divergence
\cite{Renyi,OP,Petz}.
On the other hand, for rates above $D(\rho_1\|\sigma_1)$, we see a \ki{strong converse} behaviour; namely, the type I error not only does not vanish asymptotically, but it goes to $1$ exponentially fast, and the best achievable exponent is
\begin{align}
sc(r|\rho\|\sigma)&:=\inf\left\{-\liminf_{n\to+\infty}\frac{1}{n}\log(1-\alpha_n(T_n)):\right.\nonumber\\
&\ds\ds\ds\ds\ds\ds\left.\,\limsup_{n\to+\infty}\frac{1}{n}\log\beta_n(T_n)<-r\right\}\label{thm:iid sc rate1}\\
&=H_r^*(\rho_1\|\sigma_1):=
\sup_{1<\alpha}\frac{\alpha-1}{\alpha}\left[r-D_{\alpha}\nw(\rho_1\|\sigma_1)\right].\label{thm:iid sc rate}
\end{align}
Here, $sc(r|\rho\|\sigma)$ is the \ki{strong converse exponent} of the problem, $H_r^*(\rho_1\|\sigma_1)$ is the \ki{Hoeffding anti-divergence}, and
$D_{\alpha}\nw(\rho_1\|\sigma_1):=\frac{1}{\alpha-1}\log\Tr(\rho_1^{1/2}\sigma_1^{(1-\alpha)/\alpha}\rho_1^{1/2})^{\alpha}$ is an alternative version of the quantum R\'enyi
divergence, recently introduced in \cite{Renyi_new,WWY}.

The expression for the direct exponent in the classical case (corresponding to commuting $\rho_1$ and $\sigma_1$) has been obtained
in \cite{Hoeffding,Blahut}. The exponential decay of the type I error probabilities in the non-commuting case has been proved in \cite{HO}, with bounds on
$d(r|\rho\|\sigma)$ similar to \eqref{thm:iid direct rate} in form but with a different quantum version of the R\'enyi divergences.
The correct form of the direct exponent has been obtained in \cite{Hayashi,Nagaoka}, based on techniques developed for the quantum
Chernoff bound in
\cite{Aud,NSz,ANSzV}.
The strong converse exponent in the classical case has been first determined in \cite{HK}, expressed as an optimization of relative entropies.
The strong converse property of the quantum Stein's lemma has been first proved in \cite{ON}, with a suboptimal bound on
the strong converse exponent;
the proof was later much simplified in \cite{Nagaoka2}.
An expression for the strong converse exponent in quantum hypothesis testing was given in \cite{H:text},
with an asymptotic post-measurement version of the $D_{\alpha}$ R\'enyi divergences in place of $D_{\alpha}\nw(\rho_1\|\sigma_1)$.
The expression \eqref{thm:iid sc rate} has been obtained recently in \cite{MO}.

It is worth noting that two different notions of quantum R\'enyi divergence are needed to completely describe the trade-off curve.
Indeed, the direct exponent is expressed in terms of the $D_{\alpha}\old$ divergences with $\alpha\in(0,1)$, whereas the strong
converse exponent is a function of the $D_{\alpha}\nw$ divergences with $\alpha>1$. A more explicit operational interpretation of
these divergences as generalized cutoff rates \cite{Csiszar} has been given in \cite{MH,MO}.

These results give a complete description of the trade-off between the exponents of the two error probabilities
in the i.i.d.~(independent and identically distributed) case described above.
In reality, however, there may be correlations between the different copies. That is, while the state of the individual systems
are still described by $\rho_1$ if $H_0$ is true, the global state of $n$ copies may not be of the product form
$\rho_1^{\otimes n}$ as assumed above, and similarly, $\sigma_n$ might contain correlations among the different copies.
A physically relevant example of such a scenario is where the consecutive copies are parts of a chain of particles governed
by the translates of some local Hamiltonian, and $\rho_n$ and $\sigma_n$ are temperature (Gibbs) states of two different
Hamiltonians.

Exact trade-off formulas for such states have been obtained in the direct domain in \cite{HMO2}, where two general methods have been
developed for the hypothesis testing of correlated states. The first method works if the states corresponding to both hypotheses
satisfy a certain \ki{factorization property}, which was shown to be satisfied by temperature states of translation-invariant
finite-range Hamiltonians on a spin chain in \cite{HMO}. The other method requires the existence of the regularized R\'enyi
divergences $\ol D_{\alpha}\old(\rho\|\sigma):=\lim_n(1/n)D_{\alpha}\old(\rho_n\|\sigma_n)$ for all
$\alpha\in(0,1)$, and differentiability of the limit in the parameter $\alpha$. Typical examples of such states include classical irreducible Markov chains \cite{DZ}, certain finitely
correlated states \cite{FNW,HMO2}, and temperature states of non-interacting fermions \cite{MHOF} and bosons \cite{M} on a cubic
lattice. In both cases, the trade-off formula is a direct generalization of \eqref{thm:iid direct rate1}--\eqref{thm:iid direct rate}, with $\ol D_{\alpha}(\rho\|\sigma)$ in place of $D_{\alpha}(\rho_1\|\sigma_1)$.

Here we show analogous results for the trade-off in the strong converse region.
Namely, we show the following extension of the i.i.d.~result \eqref{thm:iid sc rate1}--\eqref{thm:iid sc rate}:
\begin{align}
&sc(r|\rho\|\sigma)=H_r^*(\rho\|\sigma):=
\sup_{1<\alpha}\frac{\alpha-1}{\alpha}\left[r-\ol D_{\alpha}(\rho\|\sigma)\right]\label{main result}\\
&\text{with}\ds\ds\ol D_{\alpha}(\rho\|\sigma):=\lim_n\frac{1}{n}D_{\alpha}\nw(\rho_n\|\sigma_n),\s\alpha>1,
\end{align}
if one of two conditions is satisfied:
(1) the states corresponding to both hypotheses satisfy the factorization property, or
(2) the limit $\ol D_{\alpha}(\rho\|\sigma)$ exists for all $\alpha>1$, it is a differentiable function of $\alpha$, and it coincides
with a variant of this limit, explained later.
The main examples satisfying the first condition are again the temperature states
of translation-invariant finite-range Hamiltonians on a spin chain. The second condition is satisfied by the same class of classical Markov states and finitely correlated states as in the previous paragraph,
and for temperature states of non-interacting fermions on a cubic
lattice.

The structure of the paper is as follows. In Section \ref{sec:prel} we list some mathematical preliminaries.
In Section \ref{sec:Renyi def} we summarize some properties of the quantum R\'enyi divergences,
in Section \ref{sec:asymp} we introduce their asymptotic versions,
which is needed when dealing with correlated states, and in Section \ref{sec:LF}
we introduce the Hoeffding anti-divergence based on the asymptotic R\'enyi $\alpha$-divergences.

Section \ref{sec:sc} is the main contribution of the paper. Here we start with a more general approach than described above, namely, we consider the strong converse exponent of the hypothesis testing problem between two general sequences of states
$\{\rho_n\}_{n\in\bN}$ and $\{\sigma_n\}_{n\in\bN}$, where $\rho_n$ and $\sigma_n$ are states on the same Hilbert space,
but states with different indices are not assumed to be related in any particular way. This general approach originates from the information spectrum method \cite{Han:book}.
Expressions for the strong converse exponent in the classical \cite{Han:sc} and in the quantum case \cite{NH}
were obtained in the information spectrum framework in terms of the exponents of the Neyman-Pearson tests.
Here we impose extra conditions on the two sequences in order to obtain the more explicit expression
\eqref{main result} for the strong converse exponent, which can be further evaluated for various classes of correlated
states with physical relevance. We also note that while in the information spectrum method all exponents are evalutated in terms of the asymptotics of the error probabilities along the Neyman-Pearson tests, here we also consider variants of the Neyman-Pearson tests, e.g., in Sections \ref{sec:differentiable} and \ref{sec:quasifree sc rate}.

In the beginning of Section \ref{sec:sc}, we start with two general observations.
First, a straightforward generalization of the results of \cite{ON,Nagaoka2}, utilizing the monotonicity of the R\'enyi divergences under measurements, yields $sc(r|\rho\|\sigma)\ge H_r^*(\rho\|\sigma)$ (with a slightly more general definition of $\ol D_{\alpha}(\rho\|\sigma)$); this is the content of Lemma \ref{lemma:sc lower bound}.
Next, we show in Theorems \ref{thm:exponent gen}--\ref{thm:exponent} that
if a parametric family of sequences of tests exists with certain properties then the inequality can be reversed, and
$sc(r|\rho\|\sigma)= H_r^*(\rho\|\sigma)$ holds. We then consider two general cases in which the existence of such tests can be verified.

In Section \ref{sec:differentiable} we show that if a certain classicalization of the problem, corresponding to a suitably chosen auxiliary sequence $\{\what\sigma_n\}_{n\in\bN}$, yields the same asymptotic R\'enyi divergences
as the original problem, which also satisfy some regularity condition (differentiability in $\alpha$) then
the Neyman-Pearson tests of the classicalized problem can be used to fulfill the conditions of Theorem \ref{thm:exponent} and
obtain \eqref{main result}. We show in Section \ref{sec:quasifree sc rate} that the hypothesis testing problem for
gauge-invariant fermionic quasi-free states satisfies these conditions; moreover, the asymptotic R\'enyi divergences can be
explicitly expressed in terms of the symbols of the two states.
The standard classical example for which these conditions hold is the hypothesis testing of irreducible Markov chains; an expression for the strong converse exponent of this problem has been determined in \cite{NKMarkov}. We explain in Appendix
\ref{sec:classical} how the results of \cite{NKMarkov} can be obtained from our general considerations.

In Section \ref{sec:fact} we consider a special class of states on an infinite spin chain, satisfying a
certain factorization property. We evaluate the exponents of the type I success- and the type II error probabilities corresponding to the standard Neyman-Pearson tests, and show that these tests
satisfy the conditions of Theorem \ref{thm:exponent}, from which we can conclude that \eqref{main result} holds.
Once these exponents are available, \eqref{main result} can also be obtained from the general information spectrum formula for the strong converse exponent, given in \cite[Theorem 4]{NH}, as we explain in Remark \ref{rem:NH}.
The factorization property is known to hold for the Gibbs states of finite-range translation-invariant Hamiltonians \cite{HMO}, as we discuss in Section \ref{sec:Gibbs}.

We remark that neither the conditions of Section \ref{sec:differentiable} nor the factorization property of
Section \ref{sec:fact} need to hold in the classical case, i.e., for commuting $\rho_n$ and $\sigma_n$. Hence, our results may have non-trivial applications even for classical hypothesis testing.

Background material on classical large deviations, fermionic quasi-free states
and Szeg\H o's theorem is given in the Appendices.

\section{Preliminaries}
\label{sec:prel}

For a finite-dimensional Hilbert space $\hil$, let $\B(\hil)$ denote the set of linear operators on $\hil$, let $\B(\hil)_+$ denote the set of non-zero positive semidefinite operators,
$\B(\hil)_{++}$ the set of positive definite operators on $\hil$,
and let $\S(\hil):=\{\rho\in\B(\hil)_+:\,\Tr\rho=1\}$ be the set of
\ki{density operators} or \ki{states}.

We call a map $\map:\,\B(\hil)\to\B(\kil)$ a \ki{positive map} if $\map$ is linear, and
$\map\bz\B(\hil)_+\jz\subseteq\B(\kil)_+$.
%For every finite-dimensional Hilbert space $\hil$
For every finite-dimensional Hilbert space $\hil$, $(X,Y)\mapsto \Tr X^*Y$ is an inner product on $\B(\hil)$ (called the Hilbert-Schmidt inner product),
and for a linear map $\map:\,\B(\hil)\to\B(\kil)$, we denote its adjoint with respect to the Hilbert-Schmidt inner products on $\B(\hil)$ and $\B(\kil)$ by $\map^*$.
It is easy to see that $\map$ is positive if and only if $\map^*$ is positive, and $\map$ is trace-preserving if and only if $\map^*$ is unital.

For a self-adjoint operator $X$ on a finite-dimensional Hilbert space, let $\{X\ge 0\}$ denote the spectral projection of $X$ corresponding to the non-negative eigenvalues of $X$. The spectral projections $\{X>0\}$, $\{X\le 0\}$ and $\{X<0\}$ are defined similarly. The positive part $X_+$ of $X$ is defined as $X_+:=X\{X>0\}$. It is easy to see that
\begin{align}\label{positive part}
\Tr X_+=\max\{\Tr XT:\,0\le T\le I\},
\end{align}
a fact that we will use without further notice.

\begin{lemma}\label{lem:pos-tr-mono}
For any Hermitian operators $A,B\in\B(\hil)$,
\begin{align}\label{lem:pos-tr-mono1}
A\ge B\ds\imp\ds\Tr A_+ \ge \Tr B_+.
\end{align}
For any Hermitian operator $A$ and any positive trace preserving map $\F$, we have
\begin{align}\label{lem:pos-tr-mono2}
\Tr A_+ \ge \Tr \F(A)_+
\end{align}
\end{lemma}
\begin{proof}
The first assertion follows from
\begin{align*}
\Tr B_+ = \Tr B\{B>0\} \le \Tr A\{B>0\} \le\Tr A\{A>0\},
\end{align*}
where the first inequality is due to the assumption $B\le A$, and the second is due to \eqref{positive part}.
The second assertion follows by
\begin{align*}
\Tr \F(A)_+ &= \Tr \F(A)\{\F(A)>0\} = \Tr A \F^*(\{\F(A)>0\})\\
& \le \Tr A \{A>0\},
\end{align*}
where $\F^*$ is the Hilbert-Schmidt adjoint of $\F$, and we used that $\F^*$ is positivity preserving and unital.
\end{proof}

We will follow the convention that powers of a positive semidefinite operator $A$ are taken on its support only, and defined to
be $0$ on the orthocomplement of its support. That is, if $\lambda_1,\ldots,\lambda_r$ are the strictly positive eigenvalues of
$A$ with corresponding spectral projections $P_1,\ldots,P_r$, then
$A^t:=\sum_{i=1}^r\lambda_i^t P_i$.
In particular, $A^0$ denotes the projection onto the support of $A$,
and $A^0\le B^0$ is a shorthand for $\supp A\subseteq \supp B$ when $A,B\in\B(\hil)_+$.
Similarly, we define $\log A$ to be $0$ on the orthocomplement of $A$.

For an operator $\sigma\in\B(\hil)$, we denote by $v(\sigma)$ the number of different eigenvalues of $\sigma$.
If $\sigma$ is self-adjoint with spectral projections $P_1,\ldots,P_r$, then the \ki{pinching} by $\sigma$
is the map $\E_{\sigma}:\,\B(\hil)\to\B(\hil)$, defined as
\begin{align*}
\E_{\sigma}:\,X\mapsto \sum_{i=1}^r P_i X P_i,\ds\ds\ds X\in\B(\hil).
\end{align*}
The \ki{pinching inequality} \cite{H:pinching,H:text} tells that if $X$ is positive semidefinite then
\begin{align}\label{pinching inequality}
X\le v(\sigma)\E_{\sigma}(X).
\end{align}

\section{R\'enyi divergences and related quantities}
\label{sec:Renyi}

\subsection{Definitions and general properties}
\label{sec:Renyi def}

For non-zero positive semidefinite operators $\rho,\sigma$ on a finite-dimensional Hilbert space $\hil$, let
\begin{align*}
&Q_t\old(\rho\|\sigma):=\Tr\rho^t\sigma^{1-t},
&\pfo{t}{\rho}{\sigma}:=\log Q_t\old(\rho\|\sigma)
\end{align*}
for any $t\in\bR$, and
\begin{align*}
&Q_t\nw(\rho\|\sigma):=\Tr\bz \rho^{\half}\sigma^{\frac{1-t}{t}}\rho^{\half}\jz^t,
&\pfn{t}{\rho}{\sigma}:=\log Q_t\nw(\rho\|\sigma),
\end{align*}
for any $t>0$.
In the following, let $\xx$ denote either $\neww$ or $\oldd$, where $\oldd$ stands for the empty string. That is,
$Q_t\x(\rho\|\sigma)$ with $\xx=\oldd$ is simply $Q_t(\rho\|\sigma)$.
When $\rho$ and $\sigma$ commute, the expressions with and without $\neww$ coincide, and therefore we omit
$\neww$ in the notation.

The \ki{R\'enyi $\alpha$-divergences} of $\rho$ w.r.t.~$\sigma$ for parameter $\alpha\in[0,+\infty)\setminus\{1\}$ are defined as
\begin{align*}
&D_{\alpha}\x(\rho\|\sigma):=\lim_{\ep\searrow 0}\frac{1}{\alpha-1}\pfx{\alpha}{\rho}{\sigma+\ep I}
-\frac{1}{\alpha-1}\log\Tr\rho\\
&=
\begin{cases}
\frac{1}{\alpha-1}\pfx{\alpha}{\rho}{\sigma}-\frac{1}{\alpha-1}\log\Tr\rho,&\rho^0\le\sigma^0\\
 & \text{\s or \s}\alpha\in(0,1),\\
+\infty,&\text{otherwise}.
\end{cases}
\end{align*}
The equality above is straightforward to verify for $\xx=\oldd$, and it follows from Lemma 12 in \cite{Renyi_new} for
$\xx=\neww$.
For $\alpha=1$ we define
\begin{align}\label{relentropy def}
\rsr{\rho}{\sigma}{1}&:=\lim_{\alpha\to 1}\rsrx{\rho}{\sigma}{\alpha}\nonumber\\
&=
D(\rho\|\sigma)\nonumber\\
&:=
\begin{cases}
\frac{1}{\Tr\rho}\left[\Tr\rho\log\rho-\Tr\rho\log\sigma\right],&\rho^0\le\sigma^0,\\
+\infty,&\text{otherwise}.
\end{cases}
\end{align}
Note that $D(\rho\|\sigma)$ is the \ki{relative entropy} \cite{Umegaki,Wehrl,OP} of $\rho$ w.r.t.~$\sigma$. The above limit relation
for $D_{\alpha}$ is straightforward to verify, and for $D_{\alpha}^{*}$ it has been shown by different methods in
\cite{Renyi_new,WWY,Mosonyi}. We also provide a proof for it below.
It has been shown in \cite[theorem 5]{Renyi_new} that
\begin{align}\label{dmax}
\rsrn{\rho}{\sigma}{\infty}&:=\lim_{\alpha\to+\infty}\rsrn{\rho}{\sigma}{\alpha}\\
&=
D_{\max}(\rho\|\sigma):=
\inf\{\gamma:\,\rho\le e^{\gamma}\sigma\},
\end{align}
where $D_{\max}(\rho\|\sigma)$ is the \ki{max-relative entropy} of $\rho$ w.r.t.~$\sigma$ \cite{RennerPhD,Datta}.

\begin{lemma}\label{lemma:derivative}
Let $\rho$ and $\sigma$ be such that $\rho\sigma\ne 0$. Then $\pfx{.}{\rho}{\sigma}$ is differentiable on $(0,+\infty)$, and
\begin{align}
\frac{d}{dt}\pfo{t}{\rho}{\sigma}&=
\frac{1}{Q_t\old(\rho\|\sigma)}\Tr\rho^t\sigma^{1-t}(\log\rho-\log\sigma),\label{old psi derivative}\\
\frac{d}{dt}\pfn{t}{\rho}{\sigma}&=\frac{1}{Q_t\nw(\rho\|\sigma)}
\left[\Tr\bz \rho^{\half}\sigma^{\frac{1-t}{t}}\rho^{\half}\jz^t\log \bz\rho^{\half}\sigma^{\frac{1-t}{t}}\rho^{\half}\jz\right.\nonumber\\
&\left.-\frac{1}{t}\Tr\bz \rho^{\half}\sigma^{\frac{1-t}{t}}\rho^{\half}\jz^{t-1}
\rho^{\half}\sigma^{\frac{1-t}{t}}(\log\sigma)\rho^{\half}\right].\label{new psi derivative}
\end{align}
Moreover, \eqref{old psi derivative} is valid for all $t\in\bR$.
\end{lemma}
\begin{proof}
The derivative of $\psi\old$ is straightforward to compute.
To see the derivative of $\psi\nw$,
first note that for every $t>0$, there exist $c_t,d_t>0$ such that $c_t\rho^{\half}\sigma^{0}\rho^{\half}\le
\rho^{\half}\sigma^{\frac{1-t}{t}}\rho^{\half}\le d_t\rho^{\half}\sigma^{0}\rho^{\half}$, and thus
$(\rho^{\half}\sigma^{\frac{1-t}{t}}\rho^{\half})^0=(\rho^{\half}\sigma^{0}\rho^{\half})^0=:P$, independently of $t$.
Hence, $\rho^{\half}\sigma^{\frac{1-t}{t}}\rho^{\half}$ can be seen as an invertible positive operator on $\ran P$.
Define
\begin{align*}
&g:\,\bR_{++}\to\bR_{++}\oplus\B(\ran P)_{++},\ds\ds\ds
g(t):= t\oplus\rho^{\half}\sigma^{\frac{1-t}{t}}\rho^{\half},\ds\ds\ds\text{and}\\
&f:\,\bR_{++}\oplus\B(\ran P)_{++}\to\bR,\ds\ds\ds\ds\s
f(t\oplus X):=\Tr X^t,
\end{align*}
where $\bR_{++}=(0,+\infty)$,
so that $\Tr\bz \rho^{\half}\sigma^{\frac{1-t}{t}}\rho^{\half}\jz^t=f(g(t))$. Then $g$ has derivative
\begin{equation*}
\frac{d}{dt}g(t)=1\oplus\bz-\frac{1}{t^2}\jz\rho^{\half}\sigma^{\frac{1-t}{t}}(\log\sigma)\rho^{\half},
\end{equation*}
and the derivative of $f$ at a point $(t,X)$ is the linear map
\begin{equation}\label{f der}
df(t,X):\,(s,Y)\mapsto s\Tr X^t\log X+t\Tr X^{t-1}Y.
\end{equation}
The second term in \eqref{f der} can be obtained e.g.~from Theorem V.3.3 in \cite{Bhatia}.
Using the chain rule for derivatives, we get that
\begin{align*}
&\frac{d}{dt}\Tr\bz \rho^{\half}\sigma^{\frac{1-t}{t}}\rho^{\half}\jz^t
=
\frac{d}{dt}f(g(t))\\
&=
\Tr\bz \rho^{\half}\sigma^{\frac{1-t}{t}}\rho^{\half}\jz^t
\log\bz \rho^{\half}\sigma^{\frac{1-t}{t}}\rho^{\half}\jz\\
&\ds-\frac{1}{t}\Tr\bz \rho^{\half}\sigma^{\frac{1-t}{t}}\rho^{\half}\jz^{t-1}
\rho^{\half}\sigma^{\frac{1-t}{t}}(\log\sigma)\rho^{\half},
\end{align*}
which yields \eqref{new psi derivative}.
\end{proof}

\begin{cor}
The limit relation in \eqref{relentropy def} holds.
\end{cor}
\begin{proof}
Assume first that $\rho^0\le\sigma^0$. Then we have $D_{\alpha}\x(\rho\|\sigma)=\frac{\psi\x(\alpha|\rho\|\sigma)-\psi\x(1|\rho\|\sigma)}{\alpha-1}$, and hence
\begin{align*}
\lim_{\alpha\to 1}D_{\alpha}\x(\rho\|\sigma)=\frac{d}{d\alpha}\psi\x(\alpha|\rho\|\sigma)\Big\vert_{\alpha=1}=
D(\rho\|\sigma),
%\frac{1}{\Tr\rho}\left[\Tr\rho\log\rho-\Tr\rho\log\sigma\right],
\end{align*}
where the last equality is due to lemma \ref{lemma:derivative}. Assume next that $\rho^0\nleq\sigma^0$.
Then $D_{\alpha}\x(\rho\|\sigma)=+\infty=D(\rho\|\sigma)$ for every $\alpha>1$. On the other hand,
$\psi\x(1|\rho\|\sigma)=\log\Tr\rho\sigma^0<\log\Tr\rho$, and hence
for $\alpha<1$ we have
\begin{align*}
D_{\alpha}\x(\rho\|\sigma)=&
\frac{\psi\x(\alpha|\rho\|\sigma)-\psi\x(1|\rho\|\sigma)}{\alpha-1}\\
&+\frac{\log\Tr\rho\sigma^0-\log\Tr\rho}{\alpha-1}.
\end{align*}
The first term has a finite limit, again due to lemma \ref{lemma:derivative}, while the second term goes to
$+\infty=D(\rho\|\sigma)$ as $\alpha\nearrow 1$.
\end{proof}
%\medskip

It is easy to see (by simply computing its second derivative) that for fixed $\rho,\sigma\in\B(\hil)_+$, the function
$\alpha\mapsto \psi\old(\alpha|\rho\|\sigma)$ is convex on $\bR$. We have the following:
\begin{lemma}
Let $\rho,\sigma\in\B(\hil)_+$ be such that $\rho^0\le\sigma^0$. For every $\alpha>1$,
\begin{align}\label{pinching limit}
\pfn{\alpha}{\rho}{\sigma}=\lim_{n\to+\infty}\frac{1}{n}\psi(\alpha|\E_{\sigma^{\otimes n}}\rho^{\otimes n}\|\sigma^{\otimes n}),
\end{align}
where $\E_{\sigma^{\otimes n}}$ is the pinching by $\sigma^{\otimes n}$.
In particular, $\alpha\mapsto \psi\nw(\alpha|\rho\|\sigma)$ is convex.
\end{lemma}
\begin{proof}
The limit relation \eqref{pinching limit} is due to Theorem III.7 in \cite{MO}. By \eqref{pinching limit},
$\psi\nw(.|\rho\|\sigma)$ is the pointwise limit of convex functions, and hence itself is convex
on $(1,+\infty)$.
\end{proof}

\begin{cor}
Let $\rho,\sigma\in\B(\hil)_+$ be such that $\rho^0\le\sigma^0$.
Then $\alpha\mapsto D_{\alpha}\x(\rho\|\sigma)$ is monotone increasing on $(1,+\infty)$,
and
\begin{align}\label{D1 is inf}
D_1(\rho\|\sigma)=\inf_{\alpha>1}D_{\alpha}\x(\rho\|\sigma).
\end{align}
\end{cor}
\begin{proof}
 Note that $\rho^0\le\sigma^0$ implies that
$D_{\alpha}\x(\rho\|\sigma)=\frac{\psi\x(\alpha|\rho\|\sigma)-\psi\x(1|\rho\|\sigma)}{\alpha-1}$, and hence
convexity of $\psi\x(\alpha|\rho\|\sigma)$ in $\alpha$ yields that
$\alpha\mapsto D_{\alpha}\x(\rho\|\sigma)$ is a monotone increasing function of $\alpha$; in particular,
\eqref{D1 is inf} holds.
\end{proof}
\medskip

The $D_{\alpha}$ R\'enyi divergences are known to be monotone non-increasing under completely positive trace-preserving maps for
$\alpha\in[0,2]$. Monotonicity for the $D_{\alpha}^*$ R\'enyi divergences has been proved for different ranges of $\alpha$ and with different methods in
\cite{Beigi,Hiai,FL,MO,Renyi_new,WWY}:
\begin{lemma}\label{lemma:monotonicity}
Let $\rho,\sigma\in\B(\hil)_+$ and $\map:\,\B(\hil)\to\B(\kil)$ be a linear completely positive trace-preserving map. Then
\begin{align*}
D_{\alpha}\nw(\map(\rho)\|\map(\sigma))\le D_{\alpha}\nw(\rho\|\sigma),\ds\ds\ds\alpha\in[1/2,+\infty].
\end{align*}
\end{lemma}
\medskip

The following lemma is straightforward to verify:
\begin{lemma}
Let $\rho,\sigma\in\B(\hil)_+$ and $\lambda,\kappa>0$. For every $\alpha\in[0,+\infty]$,
\begin{align}
\psi\x(\alpha|\lambda\rho\|\kappa\sigma)&=\alpha\log\lambda+(1-\alpha)\log\kappa+\psi\x(\alpha|\rho\|\sigma),
\label{psi scaling}\\
D_{\alpha}\x(\lambda\rho\|\kappa\sigma)&=\log\lambda-\log\kappa+D_{\alpha}\x(\rho\|\sigma).
\label{D scaling}
\end{align}
\end{lemma}

\subsection{Asymptotic R\'enyi quantities}
\label{sec:asymp}

For every $n\in\bN$, let $\hil_n$ be a finite-dimensional Hilbert space, let
$\rho_n\in\S(\hil_n)$ be a state, and $\sigma_n\in\B(\hil_n)_+$ be a positive semidefinite operator.
These will play the role of the null- and the alternative hypotheses in the later sections. Note that we don't require the
$\sigma_n$ to be normalized; the reason is that this more general case can be treated the same way as the normalized case, and it turns out to be useful e.g., in state compression
(see, e.g., \cite{Mosonyi}).
We will use the notation
\begin{align*}
\rho:=\{\rho_n\}_{n\in\bN}\ds\ds\ds\text{and}\ds\ds\ds\sigma:=\{\sigma_n\}_{n\in\bN}.
\end{align*}
We will assume throughout that
\begin{align*}
\supp\rho_n\subseteq\supp\sigma_n,\ds n\in\bN,\ds\ds\ds \text{which we abbreviate as }
\end{align*}
\begin{align*}
\supp\rho\subseteq\supp\sigma.
\end{align*}

We will also consider an additional sequence $\what\sigma=\{\what\sigma_n\}_{n\in\bN}$ such that
\begin{align}\label{sigma hat}
\sigma_n\le\what \sigma_n,\ds\ds\ds\text{and}\ds\ds\ds (\sigma_n)^0=(\what\sigma_n)^0,\ds\ds\ds n\in\bN.
\end{align}
%{\color{red} following \cite[Theorem 14]{TH}.}
This sequence will be specified later, depending on the concrete problem. Given the sequence $\what\sigma$, we introduce
\begin{align*}
\what\rho_n:=\E_{\what\sigma_n}(\rho_n),
\end{align*}
the pinching of $\rho_n$ by $\what\sigma_n$. By the pinching inequality
\eqref{pinching inequality}, we have
\begin{equation}\label{pinching n}
\rho_n\le v(\what\sigma_n)\what\rho_n,
\end{equation}
where $v(\what\sigma_n)$ stands for the number of different eigenvalues of $\what\sigma_n$.

\begin{rem}
The application of the pinching technique in Quantum Information Theory goes back to \cite{HP} and \cite{H:pinching}.
The pinching of $\rho_n$ with $\what\sigma_n:=\sigma_n$ was the main tool to obtain the first expression for the strong converse exponent of i.i.d.~binary quantum state discrimination in \cite{H:text}, as well as for the expression in terms of the sandwiched
R\'enyi divergences in \cite{MO}. The key property used in these applications is the pinching inequality
\eqref{pinching n}, and that $\lim_{n\to+\infty}\frac{1}{n}\log v(\sigma_n)=0$ in the i.i.d.~case, where
$\sigma_n=\sigma_1^{\otimes n}$. This latter property, however, need not hold in the non-i.i.d.~case, and therefore pinching
with $\sigma_n$ may not be a viable way to extend results from the i.i.d. to the non-i.i.d.~setting. To circumvent this problem,
a clever way of grouping together the eigenvalues of $\sigma_n$ was introduced in \cite{TH}, which we review below in Example
\ref{ex:TH}. This results in a new reference operator $\what\sigma_n$ satisfying \eqref{sigma hat}, and with the additional
property that $\lim_{n\to+\infty}\frac{1}{n}\log v(\what\sigma_n)=0$ under much weaker conditions than i.i.d.
We will use this trick to obtain the strong converse exponent for gauge-invariant quasi-free states in Section \ref{sec:quasifree sc rate}. We are grateful to an anonymous referee for drawing our attention to this technique.
\end{rem}

For $\alpha>1$, we define the asymptotic R\'enyi quantities
\begin{align}
&\ol\psi(\alpha|\rho\|\sigma):=\limsup_{n\to+\infty}\frac{1}{n}\psi\nw(\alpha|\rho_n\|\sigma_n),\nonumber\\
&\ol D_{\alpha}(\rho\|\sigma):=\frac{1}{\alpha-1}\ol\psi(\alpha|\rho\|\sigma)=
\limsup_{n\to+\infty}\frac{1}{n}D_{\alpha}\nw(\rho_n\|\sigma_n),\label{mean psis}\\
&\what\psi(\alpha|\rho\|\sigma):=\limsup_{n\to+\infty}\frac{1}{n}\pf{\alpha}{\what\rho_n}{\what\sigma_n},\nonumber\\
&\what D_{\alpha}(\rho\|\sigma):=\frac{1}{\alpha-1}\what\psi(\alpha|\rho\|\sigma)
=
\limsup_{n\to+\infty}\frac{1}{n} D_{\alpha}(\what\rho_n\|\what\sigma_n).\label{mean psis2}
\end{align}

\begin{lemma}
$\ol\psi(\alpha|\rho\|\sigma)$ and
$\what\psi(\alpha|\rho\|\sigma)$ are convex in $\alpha$ on $(1,+\infty)$,
the functions
\begin{align}
&\alpha\mapsto\ol D_{\alpha}(\rho\|\sigma)\ds\ds\text{and}\ds\ds
\alpha\mapsto\what D_{\alpha}(\rho\|\sigma)\nonumber\\
&\text{are monotone increasing},\label{monotonicity in alpha}
\end{align}
and hence
\begin{align}
\ol D_{1}(\rho\|\sigma)&:=\inf_{\alpha>1}\ol D_{\alpha}(\rho\|\sigma)
=
\lim_{\alpha\searrow 1}\ol D_{\alpha}(\rho\|\sigma),\\
\ol D_{\infty}(\rho\|\sigma)&=\sup_{\alpha>1}\ol D_{\alpha}(\rho\|\sigma)
=
\lim_{\alpha\to+\infty}\ol D_{\alpha}(\rho\|\sigma),\label{mean dmax}\\
\what D_{1}(\rho\|\sigma)&:=\inf_{\alpha>1}\what D_{\alpha}(\rho\|\sigma)
=
\lim_{\alpha\searrow 1}\ol D_{\alpha}(\rho\|\sigma),\\
\what D_{\infty}(\rho\|\sigma)&=\sup_{\alpha>1}\what D_{\alpha}(\rho\|\sigma)
=
\lim_{\alpha\to+\infty}\what D_{\alpha}(\rho\|\sigma).%\label{pinched dmax}
\end{align}
\end{lemma}
\begin{proof}
Both $\ol\psi$ and $\what\psi$ are the limsup of convex functions, and hence are convex.
Note that $\supp\rho\subseteq\supp\sigma$ implies $\ol\psi(1|\rho\|\sigma)=0$, and hence
$\ol D_{\alpha}(\rho\|\sigma)=\frac{1}{\alpha-1}\ol\psi(\alpha|\rho\|\sigma)=\frac{1}{\alpha-1}(\ol\psi(\alpha|\rho\|\sigma)-\ol\psi(1|\rho\|\sigma))$.
From this  \eqref{monotonicity in alpha} follows for $\ol D_{\alpha}(\rho\|\sigma)$, and the proof for
$\what D_{\alpha}(\rho\|\sigma)$ goes exactly the same way.
\end{proof}

We say that $\ol\psi(\alpha|\rho\|\sigma)$ (resp., $\what\psi(\alpha|\rho\|\sigma)$) exists as a limit, if the
corresponding limsup in \eqref{mean psis}--\eqref{mean psis2} can be replaced with a limit.
We have the following:
\begin{lemma}\label{lemma:ol equals what}
For every $\alpha>1$ and $n\in\bN$,
\begin{align}
&\frac{1}{n}\psi\nw(\alpha|\rho_n\|\sigma_n)-\frac{\alpha}{n}\log v(\what\sigma_n)+\frac{1-\alpha}{n}D_{\max}(\what\sigma_n\|\sigma_n)\nonumber\\
&\ds\le
\frac{1}{n}\psi(\alpha|\what\rho_n\|\what\sigma_n)\le
\frac{1}{n}\psi\nw(\alpha|\rho_n\|\sigma_n).\label{pinching inequality for psi}
%\frac{1}{n}\psi(\alpha|\what\rho_n\|\what\sigma_n)&\le
%\frac{1}{n}\psi\nw(\alpha|\rho_n\|\sigma_n),\\
%\frac{1}{n}\psi(\alpha|\what\rho_n\|\what\sigma_n)&\ge
%\frac{1}{n}\psi\nw(\alpha|\rho_n\|\sigma_n)-\frac{2}{n}\log v(\what\sigma_n)+\frac{1-\alpha}{n}D_{\max}(\what\sigma_n\|\sigma_n).
\end{align}
In particular, if
$\lim_{n\to+\infty}\frac{1}{n}\log v(\what\sigma_n)=0=\lim_{n\to+\infty}\frac{1}{n}D_{\max}(\what\sigma_n\|\sigma_n)$
then
\begin{equation*}
\ol\psi(\alpha|\rho\|\sigma)=\what\psi(\alpha|\rho\|\sigma),
\end{equation*}
and $\ol\psi(\alpha|\rho\|\sigma)$ exists as a limit if and only if
$\what\psi(\alpha|\rho\|\sigma)$ exists as a limit.
\end{lemma}
\begin{proof}
By the monotonicity of $D_{\alpha}^*$ under pinching \cite[Proposition 14]{Renyi_new}, we have
$\psi(\alpha|\what\rho_n\|\what\sigma_n)\le \psi\nw(\alpha|\rho_n\|\what\sigma_n)$. For $\alpha>1$, the function $x\mapsto x^{\frac{1-\alpha}{\alpha}}$ is operator monotone decreasing on $(0,+\infty)$, and $X\mapsto\Tr X^{\alpha}$ is monotone increasing on positive semidefinite operators (with respect to the positive semidefinite ordering), and hence
$\sigma_n\le\what\sigma_n$ yields
$\psi\nw(\alpha|\rho_n\|\what\sigma_n)\le \psi\nw(\alpha|\rho_n\|\sigma_n)$. This proves the second inequality in
\eqref{pinching inequality for psi}.

According to the proof of \cite[Theorem 3.7]{MO}),
\begin{align*}
\psi(\alpha|\what\rho_n\|\what\sigma_n)\ge\psi\nw(\alpha|\rho_n\|\what\sigma_n) -\alpha\log v(\what\sigma_n).
\end{align*}
By \eqref{dmax}, $\what\sigma_n\le c_n\sigma_n$, where $c_n:=e^{D_{\max}(\what\sigma_n\|\sigma_n)}$. By the same monotonicity argument as above,
\begin{align*}
\psi\nw(\alpha|\rho_n\|\what\sigma_n)&\ge
\psi\nw(\alpha|\rho_n\|c_n\sigma_n)\\
 &=
\psi\nw(\alpha|\rho_n\|\sigma_n)+(1-\alpha)D_{\max}(\what\sigma_n\|\sigma_n),
\end{align*}
where the last identity is due to \eqref{psi scaling}. This proves the first inequality in
\eqref{pinching inequality for psi}.

The rest of the Lemma is obvious from \eqref{pinching inequality for psi}.
\end{proof}

The following construction is from the poof of \cite[Theorem 14]{TH}, which we review here in detail for readers' convenience:
\begin{ex}\label{ex:TH}
Let $\lambda_{1,n},\ldots,\lambda_{r_n,n}$ be the different non-zero eigenvalues of $\sigma_n$ with corresponding spectral projections $P_{1,n},\ldots,P_{r_n,n}$, and let $\lambda_{\max}(\sigma_n):=\lambda_{1,n}$,
$\lambda_{\min}(\sigma_n):=\lambda_{r_n,n}$.
%Let $\lambda_{\max,n}$ denote the largest eigenvalue of $\sigma_n$, and $\lambda_{\min,n}$ the smallest non-zero eigenvalue.
Let %$l_n:=\lfloor\log\lambda_{\max}(\sigma_n)-\log\lambda_{\min}(\sigma_n)\rfloor+1$, and
$q_n:=\lambda_{\max}(\sigma_n)/\lambda_{\min}(\sigma_n)$, and $l_n:=\lfloor\log q_n\rfloor+1$. Then for every
$i$, there exists a unique $k_i\in\{-1,\ldots,l_n-1\}$ such that
$\lambda_{\min}(\sigma_n)q_n^{\frac{k_i}{l_n}}<\lambda_{i,n}\le\lambda_{\min}(\sigma_n)q_n^{\frac{k_i+1}{l_n}}$. Define
$\what\lambda_{i,n}:=\lambda_{\min}(\sigma_n)q_n^{\frac{k_i+1}{l_n}},\,i=1,\ldots,r_n$, and
$\what\sigma_n:=\sum_{i=1}^{r_n}\what\lambda_{i,n}P_{i,n}$.
Then
\begin{align}
&v(\what\sigma_n)\le \left\lfloor\log \frac{\lambda_{\max}(\sigma_n)}{\lambda_{\min}(\sigma_n)}\right\rfloor+1
\le\lfloor-\log\lambda_{\min}(\sigma_n)\rfloor+1\label{TH0}\\
&\text{and}\ds\ds\ds
\sigma_n\le\what\sigma_n\le q_n^{\frac{1}{l_n}}\sigma_n,\label{TH1}
\end{align}
and by the last inequality,
\begin{align}\label{TH2}
D_{\max}(\what\sigma_n\|\sigma_n)\le\frac{1}{l_n}\log q_n\le 1.
\end{align}
\end{ex}

Following \cite{TH}, we introduce the notation
\begin{align*}
\theta(\sigma_n):=\min\left\{v(\sigma_n),\left\lfloor\log \frac{\lambda_{\max}(\sigma_n)}{\lambda_{\min}(\sigma_n)}\right\rfloor+1\right\}.
\end{align*}

\begin{cor}\label{cor:TH}
Consider one of the following scenarios:
\begin{enumerate}
\item
$\lim_{n\to+\infty}\frac{1}{n}\log v(\sigma_n)=0$, and we define $\what\sigma_n:=\sigma_n,\,n\in\bN$.

\item\label{ii}
$\lim_{n\to+\infty}\frac{1}{n}\log\bz\left\lfloor\log \frac{\lambda_{\max}(\sigma_n)}{\lambda_{\min}(\sigma_n)}\right\rfloor+1\jz=0$, and we define
$\{\what\sigma_n\}_{n\in\bN}$ as in Example \ref{ex:TH}.

\item\label{iii}
$\lim_{n\to+\infty}\frac{1}{n}\log\theta(\sigma_n)=0$, and
for every $n\in\bN$, if $\theta(\sigma_n)=v(\sigma_n)$ then let $\what\sigma_n:=\sigma_n$, otherwise let
$\what\sigma_n$ be the state constructed in Example \ref{ex:TH}.
\end{enumerate}
Then
\begin{align}\label{psi equality}
\ol\psi(\alpha|\rho\|\sigma)=\what\psi(\alpha|\rho\|\sigma).
\end{align}
\end{cor}
\begin{proof}
Immediate from \eqref{pinching inequality for psi}--\eqref{TH2}.
\end{proof}

\begin{cor}\label{cor:TH2}
Assume that there exist constants $c,d>0$ and $\nu\in\bR$ such that
$c^{n^{\nu}}(\sigma_n)^0\le\sigma_n\le d^{n^{\nu}}(\sigma_n)^0$ for all large enough $n$. Then the sequence $\{\what\sigma_n\}_{n\in\bN}$ constructed in Example \ref{ex:TH} satisfies
\eqref{psi equality}.
%\begin{align}\label{psi equality2}
%\ol\psi(\alpha|\rho\|\sigma)=\what\psi(\alpha|\rho\|\sigma).
%\end{align}
\end{cor}
\begin{proof}
Immediate from \ref{ii} of Corollary \ref{cor:TH}.
%\eqref{pinching inequality for psi}, \eqref{TH1}, \eqref{TH2}.
\end{proof}

\begin{rem}
A similar condition as in Corollary \ref{cor:TH2} was applied to the exponent of secret key generation in
\cite[Section VI C]{Hayashi-ld}.
\end{rem}

\subsection{Generalized Legendre transforms}
\label{sec:LF}

In this section we consider the extension of $H_r^*$ in \eqref{thm:iid sc rate} to general correlated states.
First we present a more general definition, corresponding to a general convex function $f$, and we will obtain the desired quantity by specializing to
$f=\ol\psi$.

For what follows, let $f:\,[1,+\infty)\to\bR_+$ be a non-negative convex function such that $f(1)=0$, and let
\begin{align}
f\lf(a)&:=\sup_{t> 1}\{a(t-1)-f(t)\},\ds\ds\ds a\in\bR,\label{general phi def}\\
H_{f,r}^*&:=\sup_{t> 1}\frac{r(t-1)-f(t)}{t}\nonumber\\
&=\sup_{0<s<1}\left\{sr-(1-s)f\bz\frac{1}{1-s}\jz\right\},\ds\ds\ds r\in\bR.\label{general H def}
\end{align}
Note that $a\mapsto f\lf(a)+a$ is the \ki{Legendre-Fenchel transform} (or \ki{polar transform}) of $f$ on $(1,+\infty)$, and
$r\mapsto H_{f,r}^*$ is the Legendre-Fenchel transform of $s\mapsto (1-s)f\bz\frac{1}{1-s}\jz$ on $(0,1)$.
Let
\begin{align}
D_{f,1}&:=a_{f,\min}:=\inf_{1<t<+\infty}\frac{f(t)}{t-1}=\lim_{t\searrow 1}\frac{f(t)}{t-1}:=\derright f(1),\label{amin}\\
&\text{and}\ds\ds\ds\ds
r_{f,\min}:=f\lf(a_{f,\min})+a_{f,\min},\\
D_{f,\infty}&:=a_{f,\max}:=\sup_{1<t<+\infty}\frac{f(t)}{t-1}=\lim_{t\to+\infty}\frac{f(t)}{t-1},\label{amax}\\
&\text{and}\ds\ds\ds\ds
r_{f,\max}:=f\lf(a_{f,\max})+a_{f,\max}.
\end{align}
Note that $a_{f,\min}$ is always finite, whereas $a_{f,\max}$ can be $+\infty$, in which case also $r_{f,\max}=+\infty$.

\begin{lemma}\label{lemma:converse Hoeffding repr}
For any $a\in\bR$,
\begin{align}\label{phi pos}
f\lf(a)\ge 0,\ds\ds\ds\text{and}\ds\ds\ds f\lf(a)>0\ds\iff\ds a>\derright f(1).
\end{align}
For any $r\in [0,+\infty)$, we have
\begin{align}
0\le H_{f,r}^*
=\begin{cases}
r-a_r=f\lf(a_r), & r< f\lf(a_{f,\max})+a_{f,\max}, \\
r - a_{f,\max}, & r\ge f\lf(a_{f,\max})+a_{f,\max},
\end{cases}
\label{Hr expressions}
\end{align}
where $a_r$ is the unique solution of $r-a_r=f\lf(a_r)$. Moreover,
\begin{align}\label{Hr positivity}
0< H_{f,r}^*\ds\iff \ds r>\derright f(1)=a_{f,\min}=r_{f,\min}.
\end{align}
\end{lemma}
\begin{proof}
Non-negativity of $f\lf(a)$ and $H^*_r$ are obvious from their definitions \eqref{general phi def} and \eqref{general H def} and the fact that
$f(1)=0$. Due to the convexity of $f$, $t\mapsto\frac{f(t)}{t-1}=\frac{f(t)-f(1)}{t-1}$
is monotone increasing, proving the equalities of the limits and the infimum/supremum in \eqref{amin} and \eqref{amax}.
There exists a $t> 1$ such that
$a(t-1)-f(t)>0$ if and only if $a>\inf_{t>1}\frac{f(t)}{t-1}=\derright f(1)$, proving
\eqref{phi pos}. In particular, $f\lf(a_{f,\min})=0$, and hence $a_{f,\min}=r_{f,\min}$, proving the last identity in \eqref{Hr positivity}.
Non-negativity of $H_r^*$ and the rest of \eqref{Hr positivity} follow the same way as \eqref{phi pos}.
Hence, we have to prove the identities in \eqref{Hr expressions}.

First, we consider the case $0\le r<r_{f,\max}$.
Note that $a\mapsto f\lf(a)+a$ is strictly increasing and continuous on $[0,a_{f,\max})$, and hence
for every $0\le r<r_{f,\max}$ there exists a unique $a_r$ such that $r=f\lf(a_r)+a_r$.
By definition,
\begin{equation*}
f\lf(a_r)\ge a_r (t-1)-f(t)=(t-1)(r-f\lf(a_r))-f(t),\ds\ds\ds t\ge 1,
\end{equation*}
and equality holds in the above inequality for some $t_r\in[1,+\infty)$.
Rearranging, we get
\begin{equation*}
f\lf(a_r)\ge\frac{r(t-1)-f(t)}{t},\ds\ds\ds t\ge 1,
\end{equation*}
with equality for $t_r$, and hence
\begin{equation*}
f\lf(a_r)=\max_{t\ge 1}\frac{r(t-1)-f(t)}{t}=H^*_r.
\end{equation*}

Next, assume that $r\ge r_{f,\max}$.
Since $r<+\infty$ by assumption, we only have to consider the case $r_{f,\max}<+\infty$, which implies $a_{f,\max}<+\infty$.
Note that
\begin{align}
\lim_{t\to +\infty}\frac{r(t-1)-f(t)}{t} =
r-\lim_{t\to +\infty}\frac{t-1}{t}\frac{f(t)}{t-1}
 =r-a_{f,\max}.
\end{align}
Hence it is enough to show that
\begin{align}
\frac{r(t-1)-f(t)}{t}\le r-a_{f,\max}
\end{align}
for every $t> 1$.
Note that $r\ge r_{f,\max}=f\lf(a_{f,\max})+a_{f,\max}$ implies
\begin{align}
% R(r)=
r-a_{f,\max}\ge f\lf(a_{f,\max})\ge a_{f,\max}(t-1)-f(t)
\end{align}
for every $t> 1$, from which we obtain
\begin{align}
\frac{r+f(t)}{t}\ge a_{f,\max}.
\end{align}
Thus we have
\begin{align}
r-a_{f,\max}
\ge r-\frac{r+f(t)}{t}
=\frac{r(t-1)-f(t)}{t},
\end{align}
and hence $H^*_r=r-a_{f,\max}$, as required.
\end{proof}
\medskip

We will mainly be interested in the above general Legendre-Fenchel transforms when $f=\ol\psi$ or $f=\what\psi$. In these special cases we have
\begin{align}
\phi(a)&:=\ol\psi\lf(a)=\sup_{\alpha>1}\{a(\alpha-1)-\ol\psi(\alpha|\rho\|\sigma)\}\label{phi def}\\
\what\phi(a)&:=\what\psi\lf(a)=\sup_{\alpha>1}\{a(\alpha-1)-\what\psi(\alpha|\rho\|\sigma)\},\ds\ds\ds a\in\bR,\label{phi hat def}
\end{align}
and for every $r\ge 0$,
\begin{align*}
H_r^*(\rho\|\sigma):=H_{\ol\psi,r}^*&=
\sup_{\alpha> 1}\frac{r(\alpha-1)-\ol\psi(\alpha|\rho\|\sigma)}{\alpha}\\
&=
\sup_{\alpha>1}\frac{\alpha-1}{\alpha}\left[ r-\ol D_{\alpha}(\rho\|\sigma)\right],\\
\what H_r^*(\rho\|\sigma):=H_{\what\psi,r}^*&=
\sup_{\alpha> 1}\frac{r(\alpha-1)-\what\psi(\alpha|\rho\|\sigma)}{\alpha}\\
&=
\sup_{\alpha>1}\frac{\alpha-1}{\alpha}\left[ r-\what D_{\alpha}(\rho\|\sigma)\right].
\end{align*}
Note that
\begin{align*}
D_{\ol\psi,1}=\ol D_1(\rho\|\sigma),\ds\ds
D_{\ol\psi,\infty}=\ol D_{\infty}(\rho\|\sigma),\\
%\ds\ds\ds\text{and}\ds\ds\ds
D_{\what\psi,1}=\what D_1(\rho\|\sigma),\ds\ds
D_{\what\psi,\infty}=\what D_{\infty}(\rho\|\sigma).
\end{align*}

By lemma \ref{lemma:converse Hoeffding repr}, we have
\begin{align}\label{Hr representations}
0\le H_r^*(\rho\|\sigma)
=\begin{cases}
r-a_r=\phi(a_r), & r< \phi(a_{\max})+a_{\max}, \\
r - a_{\max}, & r\ge \phi(a_{\max})+a_{\max},
\end{cases}
\end{align}
where $a_r$ is the unique solution of $r-a_r=\phi(a_r)$ and $a_{\max}=\ol D_{\infty}(\rho\|\sigma)$, and
\begin{align}\label{Hr positivity2}
0< H_r^*(\rho\|\sigma)\ds\iff \ds r>\ol D_1(\rho\|\sigma).
\end{align}
The same relations hold for $\what H_r^*(\rho\|\sigma)$ with $\what\phi,\,\what D_1(\rho\|\sigma)$ and $\what D_{\infty}(\rho\|\sigma)$ in place of $\phi,\,\ol D_1(\rho\|\sigma)$
and $\ol D_{\infty}(\rho\|\sigma)$, respectively.

We call $H_r^*(\rho\|\sigma)$ the \ki{Hoeffding anti-divergence} of $\rho$ and $\sigma$ with parameter $r$.
It differs from the (regularized) Hoeffding divergence \cite{ANSzV,Hayashi,HMO2,Nagaoka}
in two ways: first, it is based on the $D_{\alpha}^*$ R\'enyi divergences instead of $D_{\alpha}$, and second, the optimization
is over $\alpha>1$ instead of $\alpha\in(0,1)$. Due to the latter it is monotone non-decreasing under completely positive trace-preserving maps, which is the reason why we call it an anti-divergence.
\medskip

We close this section with some observations about a differentiable $f$.
Recall that by \eqref{general H def},
\begin{align}\label{general H2}
H_{f,r}^*=\sup_{t> 1}\frac{r(t-1)-f(t)}{t}=\sup_{0<s<1}\left\{sr-F(s)\right\},
\end{align}
with $F(s):=(1-s)f\bz\frac{1}{1-s}\jz$,
where the second equality is due to the change of variables
\begin{align}\label{s-t}
s:=\frac{t-1}{t}\ds\ds\ds\text{so that}\ds\ds\ds t=\frac{1}{1-s}.
\end{align}
For the rest, we will always assume that $s$ and $t$ are related as in \eqref{s-t}.
We start with the following lemma:
\begin{lemma}\label{lemma:c-preserving transformation}
Let $f:\,(1,+\infty)\to\bR$ be a convex function. Then
\begin{align*}
F:\,s\mapsto (1-s)f\bz\frac{1}{1-s}\jz\ds\text{is convex on }\ds (0,1).
\end{align*}
\end{lemma}
\begin{proof}
Since $f$ is convex, it can be written as the supremum of affine functions, i.e.,
$f(x)=\sup_{i\in \I}\{a_ix+b_i\}$, where $\I$ is some index set, and $a_i,b_i\in\bR$. Hence,
\begin{align*}
(1-s)f\bz\frac{1}{1-s}\jz&=(1-s)\sup_{i\in\I}\left\{\frac{a_i}{1-s}+b_i\right\}\\
&=
\sup_{i\in\I}\left\{a_i+b_i(1-s)\right\},
\end{align*}
which, as the supremum of affine functions, is convex in $s$.
\end{proof}

Assume for the rest that $f$ is differentiable on $(1,+\infty)$, and it is continuous at $1$.
Then $F$ is differentiable in $(0,1)$, and
\begin{align*}
F'(s)=-f\bz\frac{1}{1-s}\jz+\frac{1}{1-s}f'\bz\frac{1}{1-s}\jz=-f(t)+tf'(t).
\end{align*}
Using the assumption that $\lim_{t\searrow 1} f(t)=f(1)=0$, we get
\begin{align*}
F'(0^+)&:=\lim_{s\searrow 0}F'(s)=\lim_{t\searrow 1}(-f(t)+tf'(t))\\
&=\lim_{t\searrow 1}f'(t)=
\derright{f}(1)=D_{f,1}=a_{f,\min},\\
F'(1^-)&:=\lim_{s\nearrow 1}F'(s)=\lim_{t\nearrow +\infty}(-f(t)+tf'(t)).
\end{align*}
Convexity of $F$ guarantees that $F'$ is continuous and monotone increasing, and hence
%By \eqref{general H2},
for every $r\in(D_{f,1},F'(1^-))$, there exists an $s_r=(t_r-1)/t_r\in(0,1)$ such that
\begin{align}
r&=F'(s_r)=-f(t_r)+t_rf'(t_r)=f\lf(a_r)+a_r,\label{diff r}
\end{align}
and hence,
\begin{align}
H_{f,r}^*&=s_rr-F(s_r)=
\frac{t_r-1}{t_r}\bz-f(t_r)+t_rf'(t_r)\jz-\frac{1}{t_r}f(t_r)\nonumber\\
&=-f(t_r)+(t_r-1)f'(t_r)=f\lf(a_r),\label{diff Hr}
\end{align}
where
\begin{align*}
a_r:=f'(t_r)=r-H_{f,r}^*.
\end{align*}
Note that convexity of $f$ implies that
\begin{align*}
D_{f,\infty}&=a_{f,\max}=\sup_{1<t<+\infty}\frac{f(t)}{t-1}=\lim_{t\to+\infty}f'(t)=:f'(+\infty),
\end{align*}
and it is easy to see that $a_r\in(D_{f,1},D_{f,\infty})$. Note that
$a\mapsto f\lf(a)+a$ is convex, monotone increasing and lower semicontinuous on $\bR$, and hence
if $a_{\min}<a_{\max}$, we have
\begin{align*}
r_{f,\max}=f\lf(a_{\max})+a_{\max}=\sup_{a_{\min}<a<a_{\max}}\left\{f\lf(a)+a\right\}.
\end{align*}
Since for every $a\in(a_{\min},a_{\max})$, there exists an $s=(t-1)/t\in(0,1)$ such that
$a=f'(t),\, F'(s)=-f(t)+tf'(t)=f\lf(a)+a$, and vice versa, for every $s=(t-1)/t\in(0,1)$, we have
$a=f'(t)\in(a_{\min},a_{\max})$
and $F'(s)=f\lf(a)+a$, we see that
\begin{align}
r_{f,\max}&=\sup_{a_{\min}<a<a_{\max}}\left\{f\lf(a)+a\right\}=
\sup_{1<t<+\infty}(-f(t)+tf'(t))\nonumber\\
&=
\sup_{s\in(0,1)}F'(s)=F'(1^-).\label{rmax}
\end{align}

\section{The strong converse exponent in binary hypothesis testing}
\label{sec:sc}

For every $n\in\bN$, let $\hil_n,\,\rho_n$ and $\sigma_n$ be as in Section \ref{sec:asymp}.
As before, we assume that $\supp\rho\subseteq\supp\sigma$, i.e.,
$\supp\rho_n\subseteq\supp\sigma_n$ for every $n\in\bN$.
For every parameter $r>0$, the lower and upper \ki{strong converse exponents} $\sci(r|\rho\|\sigma)$ and $\scs(r|\rho\|\sigma)$ of the hypothesis testing problem
with null-hypothesis $\rho$ and alternative hypothesis $\sigma$ are defined as
\begin{align*}
\sci(r|\rho\|\sigma):=\inf&\left\{\liminf_{n\to+\infty}-\frac{1}{n}\log\Tr\rho_n T_n:\right.\\
&\ds\left.\limsup_{n\to+\infty}\frac{1}{n}\log\Tr\sigma_n T_n\le-r\right\},\\
\scs(r|\rho\|\sigma):=\inf&\left\{\limsup_{n\to+\infty}-\frac{1}{n}\log\Tr\rho_n T_n:\right.\\
&\ds\left.\limsup_{n\to+\infty}\frac{1}{n}\log\Tr\sigma_n T_n\le-r\right\},
\end{align*}
where the infimum is over all sequences of tests $T_n\in\B(\hil_n)_+,\,T_n\le I,\,n\in\bN$ (cf.~\eqref{thm:iid direct rate1}).
It is easy to see that $\scs(r|\rho\|\sigma)$ can be alternatively expressed as
\begin{align}
&\scs(r|\rho\|\sigma)\nonumber\\
&=
\sup\Bigl\{ R \Bigm|
\forall \{T_n\}_{n=1}^{\infty},\; 0\le T_n \le I,\,\nonumber\\
&\ds\ds\ds\ds\ds\ds\ds\limsup_{n\to\infty}\frac{1}{n}\log\Tr\sigma_nT_n \le -r\nonumber\\
&\ds\ds\ds\ds\ds\ds\ds\Rightarrow\, \liminf_{n\to\infty}\frac{1}{n}\log\Tr\rho_nT_n \le -R
\Bigr\}\nn\\
&=
\inf\Bigl\{ R \Bigm|
\exists \{T_n\}_{n=1}^{\infty}, \; 0\le T_n \le I,\nonumber\\
&\ds\ds\ds\ds\ds\ds\ds\limsup_{n\to\infty}\frac{1}{n}\log\Tr\sigma_nT_n \le -r,\nonumber\\
&\ds\ds\ds\ds\ds\ds\ds\liminf_{n\to\infty}\frac{1}{n}\log\Tr\rho_nT_n \ge -R
\Bigr\},
\label{sc rate definition}
\end{align}
and similar expressions hold for $\sci(r|\rho\|\sigma)$ as well.
\medskip

The following lemma is essentially due to \cite{Nagaoka2} and \cite{H:text}, the only difference is that we use $D_{\alpha}^*$ instead of $D_{\alpha}$.

\begin{lemma}\label{lemma:sc lower bound}
For any $r\ge 0$, we have $\sci(r|\rho\|\sigma)\ge H_r^*(\rho\|\sigma)$.
\end{lemma}
\begin{proof}
Let $T_n\in\B(\hil_n)$ be a test and let $p_n:=\bz\Tr\rho_n T_n,\Tr\rho_n (I-T_n)\jz$ and
$q_n:=\bz\Tr\sigma_n T_n,\Tr\sigma_n (I-T_n)\jz$ be the post-measurement probability distributions.
By the monotonicity of the R\'enyi divergences under measurements, we have, for any
$\alpha>1$,
\begin{align*}
\rsrn{\rho_n}{\sigma_n}{\alpha}
&\ge
\rsr{p_n}{q_n}{\alpha}\\
&\ge
\frac{1}{\alpha-1}\log\left[(\Tr\rho_n T_n)^{\alpha}(\Tr\sigma_n T_n)^{1-\alpha}\right]\\
&=
\frac{\alpha}{\alpha-1}\log\Tr\rho_n T_n-\log\Tr\sigma_n T_n,
\end{align*}
or equivalently,
\begin{align*}
\frac{1}{n}\log\Tr\rho_n T_n
\le
\frac{\alpha-1}{\alpha}\left[\frac{1}{n}\rsrn{\rho_n}{\sigma_n}{\alpha}+\frac{1}{n}\log\Tr\sigma_n T_n\right].
\end{align*}
If $\limsup_{n\to\infty}\frac{1}{n}\log\Tr\sigma_nT_n \le -r$
then
\begin{align*}
\limsup_{n\to\infty}\frac{1}{n}\log\Tr\rho_n T_n
\le
\frac{\alpha-1}{\alpha}\left[\ol D_{\alpha}(\rho\|\sigma)-r\right],\ds\ds\ds \alpha>1.
\end{align*}
Taking the infimum in $\alpha>1$, and multiplying both sides by $-1$, the assertion follows.
\end{proof}
\smallskip

It is known that the inequality in Lemma \ref{lemma:sc lower bound} holds as an equality in the i.i.d.~case \cite{MO}, and our aim is to extend this equality to various correlated scenarios. We start with the following general converse:

\begin{thm}
\label{thm:exponent gen}
Let $f:\,[0,+\infty)\to\bR$ be a convex function such that $f(1)=0$.
Assume that for every $a\in(D_{f,1},D_{f,\infty})$ there exists a sequence of tests
$0\le T_n(a)\le I_n,\,n\in\bN$, such that
\begin{align}
\limsup_{n\to\infty}\frac{1}{n}\log\Tr\sigma_n T_n(a) &\le -(f\lf(a)+a),\label{assumption1 gen}\\
\liminf_{n\to\infty}\frac{1}{n}\log\Tr\rho_n T_n(a) &\ge - f\lf(a).\label{assumption2 gen}
\end{align}
Then
\begin{align}\label{sc rate gen}
%H_r^*(\rho\|\sigma)\le
\scs(r|\rho\|\sigma)\le H_{f,r}^*,\ds\ds\ds r\ge 0.
\end{align}
\end{thm}
\begin{proof}
Due to the representation \eqref{sc rate definition} of $\scs(r|\rho\|\sigma)$ as an infimum of rates, it is sufficient to show that
for any rate $R>H_{f,r}^*$ there exists a sequence of tests $\{T_n\}_{n=1}^{\infty}$ satisfying
\begin{align}
&\limsup_{n\to\infty}\frac{1}{n}\log\Tr\sigma_nT_n \le -r
\ds\ds\ds\ds\ds\text{and}\nonumber\\
&\liminf_{n\to\infty}\frac{1}{n}\log\Tr\rho_n T_n\ge -R.\label{semioptimal tests}
\end{align}
We prove the claim by considering three different regions of $r$.
\begin{enumerate}
\item\label{gen i}
In the case $D_{f,1}< r<r_{f,\max}$,
there exists a unique $a_r\in(D_{f,1},D_{f,\infty})$ satisfying $r-a_r= f\lf(a_r)$,
and \eqref{assumption1 gen} and \eqref{assumption2 gen} yield
\begin{align*}
\limsup_{n\to\infty}\frac{1}{n}\log\Tr\sigma_n T_n(a_{r}) &\le -( f\lf(a_r)+a_r)= -r,\\
\liminf_{n\to\infty}\frac{1}{n}\log\Tr\rho_n T_n(a_{r}) &\ge - f\lf(a_{r})=H_{f,r}^*,
\end{align*}
where the last identity is due to \eqref{Hr expressions}.

\item\label{gen ii}
In the case $0\le r\le D_{f,1}$, we have $H_{f,r}^*=0$, according to
\eqref{Hr positivity}.
For any $R>0$, we can find an $a\in(D_{f,1},D_{f,\infty})$
such that $0< f\lf(a)<R$. Note that $ f\lf(a)+a>D_{f,1}\ge r$, and
\eqref{assumption1 gen} and \eqref{assumption2 gen} yield
\begin{align*}
\limsup_{n\to\infty}\frac{1}{n}\log\Tr\sigma_n T_n(a) &\le -( f\lf(a)+a)< -r,\\
\liminf_{n\to\infty}\frac{1}{n}\log\Tr\rho_n T_n(a) &\ge - f\lf(a)>-R.
\end{align*}

\item\label{gen iii} In the case $r\ge r_{f,\max}$,
we use a modification of the tests $T_n(a)$,
following the method of the proof of Theorem 4 in \cite{NH}.
For every $a,r\in\bR$, let
\begin{align*}
T_{n}(r,a):=e^{-n(r-a- f\lf(a))}T_n(a).
\end{align*}
If
$a\in(D_{f,1},D_{f,\infty})$ and $r\ge r_{f,\max}$
then
$r> f\lf(a)+a$, and hence $0\le T_{n}(r,a)\le I$, i.e., $T_{n}(r,a)$ is a test, and
\begin{align*}
&\limsup_{n\to\infty}\frac{1}{n}\log\Tr\sigma_n T_n(r,a)\\
&\ds\ds\le-r+a+ f\lf(a)-(a+ f\lf(a))= -r,\\
&\liminf_{n\to\infty}\frac{1}{n}\log\Tr\rho_n T_n(r,a)\\
&\ds\ds\ge-r+a+ f\lf(a)- f\lf(a)=-(r-a),
\end{align*}
by \eqref{assumption1 gen} and \eqref{assumption2 gen}.
Now for any $R>H_{f,r}^*=r- D_{f,\infty}$, we can find
an $a\in(D_{f,1},D_{f,\infty})$
such that $r-D_{f,\infty}<r-a<R$, and the assertion follows.
\end{enumerate}
\end{proof}

Specializing to $f=\ol\psi$ in the above Theorem yields the following:

\begin{thm}
\label{thm:exponent}
Assume that for every $a\in(\ol D_1(\rho\|\sigma),\ol D_{\infty}(\rho\|\sigma))$ there exists a sequence of tests
$0\le T_n(a)\le I_n,\,n\in\bN$, such that
\begin{align}
\limsup_{n\to\infty}\frac{1}{n}\log\Tr\sigma_n T_n(a) &\le -(\phi(a)+a),\label{assumption1}\\
\liminf_{n\to\infty}\frac{1}{n}\log\Tr\rho_n T_n(a) &\ge -\phi(a),\label{assumption2}
\end{align}
where $\phi$ is given in \eqref{phi def}.
Then
\begin{align}\label{sc rate}
\sci(r|\rho\|\sigma)=\scs(r|\rho\|\sigma)= H_r^*(\rho\|\sigma),\ds\ds\ds r\ge 0.
\end{align}
\end{thm}
\begin{proof}
Immediate from Lemma \ref{lemma:sc lower bound} and Theorem \ref{thm:exponent gen}.
\end{proof}

\begin{rem}\label{rem:two r}
The separate treatment of two different regions of $r$ values for the strong converse exponent, as in
\ref{gen i} and \ref{gen iii} in Theorem \ref{thm:exponent gen}, dates back to \cite{NKiid}, where it was noted that randomized tests are necessary for $r$ values above a critical one.
\end{rem}

\subsection{States with differentiable $\ol\psi=\what\psi$}
\label{sec:differentiable}

Now we fix a sequence $\{\what\sigma_n\}_{n\in\bN}$ satisfying \eqref{sigma hat} and,
as before, we denote by $\what\rho_n$ the pinching of $\rho_n$ by $\what\sigma_n$. Let
\begin{align}\label{pinched NP}
\what S_n(a):=\{\what\rho_n-e^{na}\what\sigma_n>0\}
\end{align}
be a Neyman-Pearson test for every $a\in\bR$ and every $n\in\bN$.

\begin{thm}\label{thm:sc rate with differentiability}
Assume that for every $\alpha>1$, $\what\psi(\alpha|\rho\|\sigma)$
exists as a limit, and $\alpha\mapsto \what\psi(\alpha|\rho\|\sigma)$ is differentiable on $(1,+\infty)$.
Then
\begin{align}
\limsup_{n\to +\infty}\frac{1}{n}\log\Tr\sigma_n \what S_n(a)
&\le
\lim_{n\to +\infty}\frac{1}{n}\log\Tr\what\sigma_n \what S_n(a)\nonumber\\
&=-(\what\phi(a)+a),\label{pinched NP rates}\\
\lim_{n\to +\infty}\frac{1}{n}\log\Tr\rho_n \what S_n(a)&=-\what\phi(a)\label{pinched NP rates2}
\end{align}
for every $a\in\big(\what D_1(\rho\|\sigma),\what D_{\infty}(\rho\|\sigma)\big)$,
where $\what\phi$ is given in \eqref{phi hat def}, and
\begin{align}\label{diff sc}
H_r^*(\rho\|\sigma)\le \sci(r|\rho\|\sigma)\le \scs(r|\rho\|\sigma)\le \what H_r^*(\rho\|\sigma),\ds\ds\ds r\ge 0.
\end{align}
If, moreover, $\ol\psi(\alpha|\rho\|\sigma)=\what\psi(\alpha|\rho\|\sigma)$ for every $\alpha>1$ then
\begin{equation}\label{differentiable sc rate}
\sci(r|\rho\|\sigma)= \scs(r|\rho\|\sigma)=H_r^*(\rho\|\sigma),\ds\ds\ds r\ge 0.
\end{equation}
\end{thm}
\begin{proof}
First, note that $\what S_n(a)=\E_{\what\sigma_n}(\what S_n(a))$, and hence
\begin{align*}
\Tr\rho_n\what S_n(a)&=\Tr\rho_n\E_{\what\sigma_n}(\what S_n(a))=
\Tr\E_{\what\sigma_n}(\rho_n)\what S_n(a)\\
&=
\Tr\what\rho_n\what S_n(a).
\end{align*}
Since $\what\rho_n$ and $\what\sigma_n$ commute, we may consider them as probability mass functions on some finite set $\X_n$, and write
\begin{align}
\Tr \sigma_n \what S_n(a)&\le\Tr \what\sigma_n \what S_n(a)=\Prob_{\what\sigma_n}\bz\{x\in\X_n:\,Y_n(x)>a\}\jz\nonumber\\
&=\mu_{n,1}\bz(a,+\infty)\jz,\label{diff error bound1}\\
\Tr \what\rho_n \what S_n(a)&=\Prob_{\what\rho_n}\bz\{x\in\X_n:\,Y_n(x)>a\}\jz\nonumber\\
&=\mu_{n,2}\bz(a,+\infty)\jz,\label{diff error bound2}
\end{align}
where $Y_n(x):=\frac{1}{n}\log\frac{\what\rho_n(x)}{\what\sigma_n(x)}$, and $\mu_{n,1}$ and $\mu_{n,2}$ are probability measures on
$\bR$, defined for any Borel subset $H$ of $\bR$ by
\begin{align*}
\mu_{n,1}(H)&:=\Prob_{\what\sigma_n}\bz\left\{x\in\X_n:\,Y_n(x)\in H\right\}\jz,\nonumber\\
\mu_{n,2}(H)&:=\Prob_{\rho_n}\bz\left\{x\in\X_n:\,Y_n(x)\in H\right\}\jz.
\end{align*}
The first inequality in \eqref{diff error bound1} is due to \eqref{sigma hat}.
Let $\lm_{n,1}$ and $\lm_{n,2}$ be the
logarithmic moment generating functions of
$\mu_{n,1}$ and $\mu_{n,2}$ respectively (see Appendix \ref{sec:ld}). Then we have
\begin{align*}
\lm_{n,1}(nt)&=\log\Exp_{\what\sigma_n}e^{t\log(\what\rho_n/\what\sigma_n)}=\log\sum_{x\in\X_n}\what\rho_n(x)^t\what\sigma_n(x)^{1-t}\\
&=\log\Tr\what\rho_n^t\what\sigma_n^{1-t}=\psi(t|\what\rho_n\|\what\sigma_n),\\
\lm_{n,2}(nt)&=\log\Exp_{\what\rho_n}e^{t\log(\what\rho_n/\what\sigma_n)}=
\log\sum_{x\in\X_n}\what\rho_n(x)^{1+t}\what\sigma_n(x)^{-t}\\
&=\log\Tr\what\rho_n^{1+t}\what\sigma_n^{-t}=\psi(1+t|\what\rho_n\|\what\sigma_n).
\end{align*}
By assumption,
\begin{align*}
\alm_1(t)&:=\lim_{n\to+\infty}\frac{1}{n}\lm_{n,1}(nt)=\what\psi(t|\rho\|\sigma),\ds\ds\ds\ds\ds\s t>1,\\
\alm_2(t)&:=\lim_{n\to+\infty}\frac{1}{n}\lm_{n,2}(nt)=\what\psi(1+t|\rho\|\sigma),\ds\ds\ds t>0.
\end{align*}
By convexity, it is easy to see that
\begin{align*}
\lim_{t\searrow 1}\alm_1'(t)=\what D_1(\rho\|\sigma),\ds\ds\text{and}\ds\ds
\lim_{t\nearrow +\infty}\alm_1'(t)\ge\what D_{\infty}(\rho\|\sigma),
\end{align*}
and, similarly,
\begin{align*}
\lim_{t\searrow 0}\alm_2'(t)=\what D_1(\rho\|\sigma),\ds\ds\text{and}\ds\ds
\lim_{t\nearrow +\infty}\alm_2'(t)\ge\what D_{\infty}(\rho\|\sigma).
\end{align*}
Using now Lemmas \ref{lemma:ldp upper bounds} and \ref{lemma:ldp lower bounds}, we get
that for every $a\in\big(\what D_1(\rho\|\sigma),\what D_{\infty}(\rho\|\sigma)\big)$,
\begin{align*}
\lim_{n\to +\infty}\frac{1}{n}\log\Tr\what\sigma_n \what S_n(a)
&=
\lim_{n\to +\infty}\frac{1}{n}\log\mu_{n,1}\bz(a,+\infty)\jz\\
&=
-\sup_{t>1}\{at-\what\psi(t|\rho\|\sigma)\}\\\
&=-(\what\phi(a)+a),\\
\lim_{n\to +\infty}\frac{1}{n}\log\Tr\rho_n \what S_n(a)
&=
\lim_{n\to +\infty}\frac{1}{n}\log\mu_{n,2}\bz(a,+\infty)\jz\\
&=
-\sup_{t>0}\{at-\what\psi(1+t|\rho\|\sigma)\}\\
&=-\what\phi(a),
\end{align*}
proving the identities in \eqref{pinched NP rates}--\eqref{pinched NP rates2}. The inequality in
\eqref{pinched NP rates} is obvious from the inequality in \eqref{diff error bound1}.

Applying Theorem \ref{thm:exponent gen} with $f:=\what\psi$ and $T_n(a):=\what S_n(a)$ yields the last inequality in \eqref{diff sc}, and the first inequality is immediate from
 Lemma \ref{lemma:sc lower bound}.
Finally, if $\ol\psi(\alpha|\rho\|\sigma)=\what\psi(\alpha|\rho\|\sigma)$ for every $\alpha>1$ then
$\what H_r^*(\rho\|\sigma)=H_r^*(\rho\|\sigma)$ for every $r$, and \eqref{diff sc} reduces to
\eqref{differentiable sc rate}.
\end{proof}

Combining Theorem \ref{thm:sc rate with differentiability} and Corollary \ref{cor:TH}, we get immediately the following:
\begin{cor}\label{cor:diff}
Assume that for every $\alpha>1$, $\ol\psi(\alpha|\rho\|\sigma)$
exists as a limit, and $\alpha\mapsto \ol\psi(\alpha|\rho\|\sigma)$ is differentiable on $(1,+\infty)$.
Assume also that $\lim_{n\to+\infty}\frac{1}{n}\log\theta(\sigma_n)=0$. Then
\begin{align}\label{differentiable sc rate2}
\sci(r|\rho\|\sigma)= \scs(r|\rho\|\sigma)=H_r^*(\rho\|\sigma),\ds\ds\ds r\ge 0.
\end{align}
Moreover, the optimal sequence of tests can be chosen as in \eqref{pinched NP}, with
$\what\sigma_n$ as in \ref{iii} of Corollary \ref{cor:TH}.
\end{cor}
\medskip

We say that the hypothesis testing problem is i.i.d. if $\hil_n=\hil_1^{\otimes n}$, $\rho_n=\rho_1^{\otimes n}$ and
$\sigma_n=\sigma_1^{\otimes n}$ for every $n\in\bN$.
Let $H_r^*(\rho_1\|\sigma_1)$ be as given in \eqref{thm:iid sc rate}.
An expression for the strong converse exponent
in the i.i.d.~case was first given in \cite{H:text}, using the tests $\what S_n(a)$ corresponding to the choice
$\what\sigma_n:=\sigma_n$. There it was shown that the inequality
\begin{align}
\scs(r|\rho\|\sigma)&\ge \what H_r^*(\rho_1\|\sigma_1)\nonumber\\
&:=
\sup_{\alpha>1}\frac{\alpha-1}{\alpha}\left[r-\lim_{n\to+\infty}\frac{1}{n}D_{\alpha}\bz\E_{\sigma_n}\rho_n\|\sigma_n\jz\right]
\label{Hayashi lower}
\end{align}
holds (cf.~Lemma \ref{lemma:sc lower bound});
the converse inequality can be obtained by applying the classical strong converse result of \cite{HK}
to the commuting states $\what\rho_n=\E_{\sigma_n}\rho_n$ and $\sigma_n$. It was shown later in
\cite{MO} that
$\lim_{n\to+\infty}\frac{1}{n}D_{\alpha}\bz\E_{\sigma_n}\rho_n\|\sigma_n\jz=D_{\alpha}^*(\rho_1\|\sigma_1)=\ol D_{\alpha}(\rho\|\sigma)$;
with this addition, \eqref{Hayashi lower} yields Lemma \ref{lemma:sc lower bound}.
The strong converse exponent was later shown to be equal to
 $H_r^*(\rho_1\|\sigma_1)$ in \cite{MO}, by showing that
\eqref{pinched NP rates}--\eqref{pinched NP rates2} hold with
$S_{n}(a):=\{\rho_n-e^{na}\sigma_n>0\}$ in place of $\what S_n(a)$.
Here we give an alternative proof, based on
Theorem \ref{thm:sc rate with differentiability}. Note that neither the proof in \cite{MO}, nor the proof below
uses the classical result as an ingredient; on the contrary, the classical result follows as a special case.

\begin{thm}\label{thm:iid}
In the i.i.d.~case, $\sci(r|\rho\|\sigma)=\scs(r|\rho\|\sigma)=H_r^*(\rho_1\|\sigma_1)$ for every $r\ge 0$.
\end{thm}
\begin{proof}
It is easy to see that the i.i.d.~assumption implies $\ol \psi(\alpha|\rho\|\sigma)=\psi\nw(\alpha|\rho_1\|\sigma_1)$, and thus also $H_r^*(\rho\|\sigma)=H_r^*(\rho_1\|\sigma_1)$.
The choice $\what\sigma_n:=\sigma_n=\sigma_1^{\otimes n}$ yields that
$v(\what\sigma_n$) grows polynomially with $n$, and $D_{\max}(\what\sigma_n\|\sigma_n)=0$, and hence, by Lemma \ref{lemma:ol equals what},
$\ol \psi(\alpha|\rho\|\sigma)=\what \psi(\alpha|\rho\|\sigma)$. Finally, differentiability of
$\ol \psi(\alpha|\rho\|\sigma)=\psi\nw(\alpha|\rho_1\|\sigma_1)$ in $\alpha$ for $\alpha>1$ follows from Lemma \ref{lemma:derivative}. Thus, all the conditions of Theorem
\ref{thm:sc rate with differentiability} are satisfied, and therefore \eqref{differentiable sc rate} holds.
\end{proof}

An expression for the strong converse exponent in the classcial i.i.d.~case was first given in \cite{HK}, followed by a different expression, based on the Hellinger arc, in \cite{NKiid}, where it was also explained how the expression in \cite{HK} can obtained from the one in  \cite{NKiid}. In Appendix \ref{sec:classical} we briefly explain how the expressions in \cite{NKiid}
can be obtained from Theorem \ref{thm:iid}.
\medskip

In \cite[Example B.1]{HMO2} a class of finitely correlated states \cite{FNW} with commutative auxiliary algebra has been studied, and it has been shown that for these states, $\lim_{n}(1/n)\psi(\alpha|\rho_n\|\sigma_n)$ is differentiable in $\alpha$
for every $\alpha\in\bR$. In particular, this class includes classical Markov chains with an irreducible transition matrix.
Exactly the same argument as in \cite{HMO2} yields that for this class of states, also $\ol \psi(\alpha|\rho\|\sigma)$ exists as a limit and is differentiable in $\alpha$ for $\alpha>1$. It is also easy to verify that if $\sigma$ is in this class then
$\lim_n(1/n)v(\sigma_n)=0$, and thus $\ol \psi(\alpha|\rho\|\sigma)=\what \psi(\alpha|\rho\|\sigma)$
due to Lemma \ref{lemma:ol equals what}.
In particular, the strong converse exponent can be expressed
as in \eqref{differentiable sc rate}, due to Theorem \ref{thm:sc rate with differentiability}.
An alternative expression for the strong converse exponent of classical Markov chains was given before in \cite{NKMarkov}. In Appendix \ref{sec:classical}, we explain how the exponent of \cite{NKMarkov} can be obtained from Theorem \ref{thm:sc rate with differentiability}.

In Section \ref{sec:quasifree sc rate} we show that Theorem \ref{thm:sc rate with differentiability} can be applied to obtain
the strong converse exponent for the hypothesis testing problem of gauge-invariant fermionic quasi-free states.

\subsection{States with factorization property}
\label{sec:fact}

Let $\hil$ be a finite-dimensional Hilbert space, and for every $n\in\bN$, let $\omega_n$ be a state on $\hil^{\otimes n}$.
We say that  $\omega:=\{\omega_n\}_{n\in\bN}$ satisfies the \ki{factorization property} if there exists an $\eta\ge 1$ such that
for every $k,m,r\in\bN$,
\begin{align*}
\omega_{km+r}&\le
\eta^{k}\omega_m^{\otimes k}\otimes\omega_{r}
& &\text{(upper factorization),\ds\ds and}\\
\omega_{km+r}&\ge \eta^{-k}\omega_m^{\otimes k}\otimes\omega_{r}& &\text{(lower factorization)}.
\end{align*}
We call $\eta$ a \ki{factorization constant} for $\omega$.
Note that if $\rho=\{\rho_n\}_{n\in\bN}$ and $\sigma=\{\sigma_n\}_{n\in\bN}$ both satisfy the factorization property
then we can always choose an $\eta$ which is a common factorization constant for both $\rho$ and $\sigma$.

Obviously, if $\omega$ is i.i.d., i.e., of the form $\omega_n=\omega_1^{\otimes n},\,n\in\bN$, then it satisfies the
factorization property with $\eta=1$. It has been shown in \cite{HMO} that finitely correlated states \cite{FNW}
satisfy the upper factorization property, but not necessarily the lower factorization property. In particular, if $\omega$ is a
classical Markov chain then it satisfies both the upper and the lower factorization property if and only if all the entries of
its transition matrix are strictly positive. Physically relevant examples of states with the factorization property are the Gibbs states of translation-invariant finite-range interactions on a spin chain; for details, see Section \ref{sec:Gibbs}.
\medskip

In this section we show that if both $\rho$ and $\sigma$ satisfy the factorization property then the tests $T_n(a):=S_n(a)$,
where
\begin{align}\label{NP def}
S_n(a):=\{\rho_n-e^{na}\sigma_n>0\}
\end{align}
are the quantum Neyman-Pearson tests,
satisfy \eqref{assumption1} and \eqref{assumption2}, and hence \eqref{sc rate} holds.
We will prove \eqref{assumption1} and \eqref{assumption2} in Lemmas \ref{lemma:np upper fact} and \ref{lem:np-lower}, and give the formal statement of our main result in Theorem \ref{thm:sc rate with factorization}.
\medskip

We start with showing that under the factorization assumption, $\ol\psi$ exists as a limit, and give bounds on its deviation
from the $\psi$ functions for finite $n$.

\begin{lemma}
Let $\rho$ and $\sigma$ satisfy the factorization property,
and let $\eta$ be a common factorization constant.
Then $\ol\psi(\alpha|\rho\|\sigma)$ exists as a limit for every $\alpha>1$, and
\begin{align}
\ol\psi(\alpha|\rho\|\sigma)-\frac{2\alpha-1}{n}\log\eta&\le \frac{1}{n}\psi\nw(\alpha|\rho_n\|\sigma_n)\nonumber\\
&\le
\ol\psi(\alpha|\rho\|\sigma)+\frac{2\alpha-1}{n}\log\eta
\label{finite-size deviations1}
\end{align}
for every $\alpha>1$ and every $n\in\bN$.
\end{lemma}
\begin{proof}
Given $m\in\bN$, every $n\in\bN$ can be uniquely written in the form $n=km+r$ with $k,r\in\bN$, $r\in\{0,\ldots,m-1\}$.
Since $\alpha>1$, we have
$-1<\frac{1-\alpha}{\alpha}<0$, and hence $x\mapsto x^{\frac{1-\alpha}{\alpha}}$ is operator monotone decreasing. Thus
\begin{align}\label{sigma factorization}
\eta^{-k\frac{\alpha-1}{\alpha}}\bz\sigma_m^{\otimes k}\otimes\sigma_r\jz^{\frac{1-\alpha}{\alpha}}
\le
\sigma_n^{\frac{1-\alpha}{\alpha}}
\le
\eta^{k\frac{\alpha-1}{\alpha}}\bz\sigma_m^{\otimes k}\otimes\sigma_r\jz^{\frac{1-\alpha}{\alpha}}.
\end{align}
Taking into account that $A\mapsto\Tr A^{\alpha}$ is monotone increasing w.r.t.~the positive semidefinite ordering, we obtain
\begin{align*}
&Q_{\alpha}\nw(\rho_n\|\sigma_n)\\
&=
\Tr\bz\rho_n^{1/2}\sigma_n^{\frac{1-\alpha}{\alpha}}\rho_n^{1/2}\jz^{\alpha}\\
&\le
\eta^{k(\alpha-1)}\Tr\bz\rho_n^{1/2}\bz\sigma_m^{\otimes k}\otimes\sigma_r\jz^{\frac{1-\alpha}{\alpha}}\rho_n^{1/2}\jz^{\alpha}\\
&=
\eta^{k(\alpha-1)}\Tr\bz\bz\sigma_m^{\otimes k}\otimes\sigma_r\jz^{\frac{1-\alpha}{2\alpha}}\rho_n\bz\sigma_m^{\otimes k}\otimes\sigma_r\jz^{\frac{1-\alpha}{2\alpha}}\jz^{\alpha}\\
&\le
\eta^{k\alpha}\eta^{k(\alpha-1)}\Tr\Big(\bz\sigma_m^{\otimes k}\otimes\sigma_r\jz^{\frac{1-\alpha}{2\alpha}}\bz\rho_m^{\otimes k}\otimes\rho_r\jz\\
&\ds\ds\ds\ds\ds\ds\ds\ds\ds\ds\ds\ds\bz\sigma_m^{\otimes k}\otimes\sigma_r\jz^{\frac{1-\alpha}{2\alpha}}\Big)^{\alpha}\\
&=
\eta^{k(2\alpha-1)}Q_{\alpha}\nw(\rho_m\|\sigma_m)^kQ_{\alpha}\nw(\rho_r\|\sigma_r),
\end{align*}
and thus
\begin{align*}
&\limsup_{n\to+\infty}\frac{1}{n}\log Q_{\alpha}\nw(\rho_n\|\sigma_n)\\
&\ds\le
\frac{2\alpha-1}{m}\log\eta+\frac{1}{m}\log Q_{\alpha}\nw(\rho_m\|\sigma_m).
\end{align*}
Taking now the liminf in $m$, we get that $\ol\psi(\alpha|\rho\|\sigma)$ exists as a limit, and  the first inequality in \eqref{finite-size deviations1} holds, for
every $\alpha>1$.
Using the lower factorization for $\rho$ and upper factorization for $\sigma$, an analogous argument to the one above yields the second inequality in \eqref{finite-size deviations1}.
\end{proof}

\begin{cor}
For every $\alpha\in(1,+\infty)$, we have
\begin{align}\label{D alpha limit}
\ol D_{\alpha}(\rho\|\sigma)=\lim_{n\to+\infty}\frac{1}{n} D_{\alpha}\nw(\rho_n\|\sigma_n).
\end{align}
\end{cor}

\begin{lemma}\label{lemma:np upper fact}
Assume that $\rho$ and $\sigma$ satisfy the factorization property. Then
\begin{align}
\limsup_{n\to\infty}\frac{1}{n}\log\Tr\rho_nS_n(a)& \le -\phi(a),\label{type I upper fact}\\
\limsup_{n\to\infty}\frac{1}{n}\log\Tr\sigma_nS_n(a)& \le -(\phi(a)+a).\label{type II upper fact}
\end{align}
for any $a\in\bR$.
\end{lemma}
\begin{proof}
First, we prove that
\begin{align}\label{ineq1 factorization}
\Tr\rho_nS_n(a)&\le e^{-na(\alpha-1)} Q_{\alpha}\nw(\rho_n||\sigma_n)
\end{align}
for every $n\in\bN$, $\alpha\ge 1$ and $n\in\bN$. Indeed, this inequality holds trivially if $S_n(a)=0$.
Otherwise we can use
\begin{align}\label{ineq2 factorization}
\Tr\rho_nS_n(a) \ge e^{na} \Tr\sigma_nS_n(a),
\end{align}
to show that for $\alpha\ge 1$,
\begin{align*}
&\Tr\rho_nS_n(a)\\
&\ds= \left\{\Tr\rho_nS_n(a) \right\}^{\alpha} \left\{ \Tr\rho_nS_n(a) \right\}^{1-\alpha} \\
&\ds\le e^{na(1-\alpha)} \left\{\Tr\rho_nS_n(a) \right\}^{\alpha} \left\{ \Tr\sigma_nS_n(a) \right\}^{1-\alpha} \\
&\ds\le e^{-na(\alpha-1)} \Big[ \left\{\Tr\rho_nS_n(a) \right\}^{\alpha} \left\{ \Tr\sigma_nS_n(a) \right\}^{1-\alpha}\\
&\ds\ds\ds+ \left\{\Tr\rho_n(I_n-S_n(a)) \right\}^{\alpha} \left\{ \Tr\sigma_n(I_n-S_n(a)) \right\}^{1-\alpha} \Big] \\
&\ds\le e^{-na(\alpha-1)} Q_{\alpha}\nw(\rho_n\|\sigma_n),
\end{align*}
where the last inequality is due to Lemma \ref{lemma:monotonicity}.
From \eqref{ineq1 factorization} we obtain
\begin{align*}
\frac{1}{n}\log\Tr\rho_nS_n(a) \le - \left\{a(\alpha-1) - \frac{1}{n}\log Q_{\alpha}\nw(\rho_n||\sigma_n)\right\},
\end{align*}
and taking first the limsup in $n$ and then the infimum over $\alpha>1$ yields \eqref{type I upper fact}.
Finally, combining \eqref{ineq2 factorization} with \eqref{type I upper fact} yields \eqref{type II upper fact}.
\end{proof}

\begin{lemma}\label{lemma:iid limit}
For any $A,B\in\B(\hil)_+$, and any $c\in(D_1(A\|B),D_{\infty}\nw(A\|B))$, we have
\begin{align*}
&\lim_{n\to+\infty}\frac{1}{n}\log\Tr\bz A^{\otimes n}-e^{nc}B^{\otimes n}\jz_+\\
&\ds=
-\sup_{\alpha>1}\{c(\alpha-1)-\psi\nw(\alpha|A\|B)\}.
\end{align*}
\end{lemma}
\begin{proof}
When $\Tr A=\Tr B=1$, the assertion follows from Theorem IV.4 in \cite{MO}. In general, let
$\tilde A:=A/\Tr A,\,\tilde B:=B/\Tr B$. Then
\begin{align}
&\Tr\bz A^{\otimes n}-e^{nc}B^{\otimes n}\jz_+\nonumber\\
&\ds=
(\Tr A)^n\Tr\bz \tilde A^{\otimes n}-e^{n(c+\log\Tr B-\log\Tr A)}\tilde B^{\otimes n}\jz_+.\label{iid rescaling}
\end{align}
By \eqref{D scaling}, we have $D_{\alpha}\nw(A\|B)=\log\Tr A-\log\Tr B+D_{\alpha}\nw(\tilde A\|\tilde B)$, and hence
$c+\log\Tr B-\log\Tr A\in(D_1(\tilde A\|\tilde B),D_{\infty}\nw(\tilde A\|\tilde B))$.
Thus, by \eqref{iid rescaling} and Theorem IV.4 in \cite{MO}, we have
\begin{align*}
&\lim_{n\to+\infty}\frac{1}{n}\log\Tr\bz A^{\otimes n}-e^{nc}B^{\otimes n}\jz_+\\
&\ds=
\log\Tr A-\sup_{\alpha>1}\{(c+\log\Tr B-\log\Tr A)(\alpha-1)\\
&\ds\ds\ds-\psi\nw(\alpha|\tilde A\|\tilde B)\}\\
&\ds=-\sup_{\alpha>1}\{c(\alpha-1)-\psi\nw(\alpha|A\|B)\},
\end{align*}
where the last equality is due to \eqref{psi scaling}.
\end{proof}

\begin{lemma}\label{lem:np-lower}
Assume that $\rho$ and $\sigma$ satisfy the factorization property. Then
%If $\rho$ and $\sigma$ satisfy both of the upper and lower factorization properties, we have
\begin{align}
\liminf_{n\to\infty}\frac{1}{n}\log\Tr\rho_nS_n(a)&
\ge \liminf_{n\to\infty}\frac{1}{n}\log\Tr(\rho_n-e^{na}\sigma_n)_+\nonumber\\
&\ge -\phi(a)
\label{lem:np-lower-1}
\end{align}
for every $a\in\left(\ol D_1(\rho\|\sigma),\ol D_{\infty}(\rho\|\sigma)\right)$.
\end{lemma}
\begin{proof}
We will assume that $\ol D_1(\rho\|\sigma)\ne\ol D_{\infty}(\rho\|\sigma)$, since otherwise the statement is empty.
Let $\eta$ denote a common factorization constant for $\rho$ and $\sigma$, and let $b\in\bR$ be such that
$\ol D_1(\rho\|\sigma)<a<b<\ol D_{\infty}(\rho\|\sigma)$. Due to \eqref{mean dmax}, there exist
$1<\alpha_1<\alpha_2<+\infty$ such that $\ol D_{\alpha_1}(\rho\|\sigma)<a<b<\ol D_{\alpha_2}(\rho\|\sigma)$.
Note that for every $\alpha>1$,
\begin{align*}
\frac{1}{m}D_{\alpha}\nw\bz\eta\inv\rho_n\|\eta\sigma_m\jz&=
-\frac{1}{m}\log\eta^2+\frac{1}{m}D_{\alpha}\nw\bz\rho_n\|\sigma_m\jz\\
&\ds\ds\xrightarrow[m\to+\infty]{}\ol D_{\alpha}(\rho\|\sigma),
\end{align*}
where the limit follows from \eqref{D alpha limit}.
Thus we see the existence of an $m_b$ such that for all $m\ge m_b$,
\begin{align}
\frac{1}{m}D_1(\eta\inv\rho_m\|\eta\sigma_m)&\le
\frac{1}{m}D_{\alpha_1}\nw\bz\eta\inv\rho_m\|\eta\sigma_m\jz\nonumber\\
&<b<
\frac{1}{m}D_{\alpha_2}\nw\bz\eta\inv\rho_m\|\eta\sigma_m\jz\nonumber\\
&\le\frac{1}{m}D_{\infty}\nw\bz\eta\inv\rho_m\|\eta\sigma_m\jz,\label{b limits}
\end{align}
where the first and the last inequalities are due to the monotonicity of the R\'enyi divergences in the parameter $\alpha$.

For a fixed $m\ge m_b$, we can write every $n>m$ uniquely as $n=km+r$ with $k\in\bN$ and $r\in\{1,\ldots,m\}$. Then we have
\begin{align}
\Tr\rho_nS_n(a)
&=\Tr(\rho_n-e^{na}\sigma_n)S_n(a)+e^{na}\Tr\sigma_nS_n(a) \nn \\
&\ge\Tr(\rho_n-e^{na}\sigma_n)_+ \nn \\
&\ge\Tr(\rho_{km}-e^{na}\sigma_{km})_+ \nn \\
&\ge\Tr(\eta^{-k}\rho_m^{\otimes k}-e^{na}\eta^{k}\sigma_m^{\otimes k})_+ ,
\label{lem:np-lower-2}
\end{align}
where the second inequality follows from the monotonicity \eqref{lem:pos-tr-mono2} applied to the partial trace
over subsystems $km+1$ to $n$,
and the last inequality is due to the factorization properties and \eqref{lem:pos-tr-mono1}.
Note that
\begin{align*}
na = (km+r)a \le kma+ma= km\left( a+\frac{a}{k} \right) < kmb ,
\end{align*}
whenever $k>a/(b-a)$, and for any such $k$ we have
\begin{align}
\Tr(\rho_n-e^{na}\sigma_n)_+ &\ge \Tr(\eta^{-k}\rho_m^{\otimes k}-e^{kmb}\eta^{k}\sigma_m^{\otimes k})_+ \nonumber\\
&=\Tr((\eta^{-1}\rho_m)^{\otimes k}-e^{kmb}(\eta\sigma_m)^{\otimes k})_+\label{factorization ineq1}
\end{align}
due to \eqref{lem:np-lower-2} and \eqref{lem:pos-tr-mono1}.
By \eqref{b limits}, $mb\in(D_1(\eta\inv\rho_m\|\eta\sigma_m),D_{\infty}\nw(\eta\inv\rho_m\|\eta\sigma_m))$, and hence
\eqref{factorization ineq1} and lemma \ref{lemma:iid limit} yield
\begin{align}
&\liminf_{n\to\infty}\frac{1}{n}\log\Tr(\rho_n-e^{na}\sigma_n)_+\nonumber\\
&\ds\ge \frac{1}{m}\liminf_{k\to\infty}\frac{1}{k}\log\Tr((\eta^{-1}\rho_m)^{\otimes k}-e^{kmb}(\eta\sigma_m)^{\otimes k})_+\nonumber\\
&\ds= -\frac{1}{m}\sup_{\alpha>1}\left\{ mb(\alpha-1)-\psi\nw(\alpha|\eta^{-1}\rho_m\|\eta\sigma_m) \right\},
\label{lem:np-lower-3}
\end{align}

By \eqref{psi scaling} and \eqref{finite-size deviations1} we have
\begin{align*}
\frac{1}{m}\psi\nw(\alpha|\eta^{-1}\rho_m\|\eta\sigma_m)
&=\frac{1}{m}\psi\nw(\alpha|\rho_m\|\sigma_m)-\frac{(2\alpha-1)}{m}\log\eta \\
&\ge\ol\psi(\alpha|\rho\|\sigma)-\frac{(4\alpha-2)}{m}\log\eta .
\end{align*}
Combining the above inequality and \eqref{lem:np-lower-3}, we have
\begin{align*}
&\liminf_{n\to\infty}\frac{1}{n}\log\Tr(\rho_n-e^{na}\sigma_n)_+\\
&\ds\ge -\sup_{\alpha>1}\left\{\left(b+\frac{4}{m}\log\eta\right)(\alpha-1)-\ol\psi(\alpha|\rho\|\sigma)\right\}
-\frac{2}{m}\log\eta \\
&\ds= -\phi\left(b+\frac{4}{m}\log\eta\right) -\frac{2}{m}\log\eta.
\end{align*}
Note that $\phi$ is continuous on $\left(\ol D_1(\rho\|\sigma),\ol D_{\infty}(\rho\|\sigma)\right)$ and
$b+\frac{4}{m}\log\eta\in\left(\ol D_1(\rho\|\sigma),\ol D_{\infty}(\rho\|\sigma)\right)$ for all sufficiently large $m$. Hence,
by taking the limit $m\to\infty$, we get
\begin{align*}
\liminf_{n\to\infty}\frac{1}{n}\log\Tr(\rho_n-e^{na}\sigma_n)_+ \ge -\phi(b).
\end{align*}
Finally, taking the limit $b\searrow a$, we get
\begin{align}
\liminf_{n\to\infty}\frac{1}{n}\log\Tr(\rho_n-e^{na}\sigma_n)_+ \ge -\phi(a) .
\label{lem:np-lower-4}
\end{align}
Now \eqref{lem:np-lower-2} and \eqref{lem:np-lower-4} lead to the assertion.
\end{proof}

\begin{thm}\label{thm:sc rate with factorization}
Assume that $\rho$ and $\sigma$ satisfy the factorization property. Then
\begin{align}
\lim_{n\to\infty}\frac{1}{n}\log\Tr\rho_nS_n(a)
&=\lim_{n\to\infty}\frac{1}{n}\log\Tr(\rho_n-e^{na}\sigma_n)_+\nonumber\\
&= -\phi(a),\label{type I limit}\\
\lim_{n\to\infty}\frac{1}{n}\log\Tr\sigma_nS_n(a)
&= -(\phi(a)+a)\label{type II limit}
\end{align}
for every $a\in\left(\ol D_1(\rho\|\sigma),\ol D_{\infty}(\rho\|\sigma)\right)$. In particular, the conditions
of Theorem \ref{thm:exponent} are satisfied with $T_n(a):=S_n(a)$, and hence
\begin{align}\label{sc for factorization}
\sci(r|\rho\|\sigma)=\scs(r|\rho\|\sigma)= H_r^*(\rho\|\sigma),\ds\ds\ds r\ge 0.
\end{align}
\end{thm}
\begin{proof}
The identities in \eqref{type I limit} are immediate from \eqref{type I upper fact} and
\eqref{lem:np-lower-1}. By \eqref{type II upper fact}, $\limsup_{n\to\infty}\frac{1}{n}\log\Tr\sigma_nS_n(a)
\le -(\phi(a)+a)$, and the same argument as in \cite[Theorem IV.5]{MO} yields
 $\liminf_{n\to\infty}\frac{1}{n}\log\Tr\sigma_nS_n(a)
\ge -(\phi(a)+a)$, proving \eqref{type II limit}. Finally, \eqref{sc for factorization}
follows from Theorem \ref{thm:exponent}.
\end{proof}

\begin{rem}\label{rem:NH}
It has been shown in \cite[Theorem 4]{NH} that for general sequences $\{\rho_n\}_{n\in\bN}$, $\{\sigma_n\}_{n\in\bN}$ as in Section \ref{sec:asymp},
\begin{align}\label{NHsc}
&\scs(r|\rho\|\sigma)=\inf_a\max\left\{r-a,\,\ol s(a) \right\},\\
&\ol s(a):=\limsup_n-\frac{1}{n}\log\Tr\rho_n S_n(a),
\end{align}
where $S_n(a)$ is given in \eqref{NP def}.
(Note that the roles of $\rho$ and $\sigma$ are reversed here as compared to \cite{NH}, which is the reason why we have $r-a$ instead of $r+a$ as in \cite[Theorem 4]{NH}.) By \eqref{type I limit},
if both $\rho$ and $\sigma$ satisfy the factorization property then
$\ol s(a)=\phi(a)$ for
$a\in\left(\ol D_1(\rho\|\sigma),\ol D_{\infty}(\rho\|\sigma)\right)$. It is easy to see that both $\ol s$ and $\phi$ are non-negative and monotone increasing. Since $\lim_{a\searrow \ol D_1(\rho\|\sigma)}\phi\bz a\jz=0$, we get
$\phi(a)=0=\ol s(a)$ for every $a\le\ol D_1(\rho\|\sigma)$. It is also clear from the definitions that
$\phi(a)=+\infty=\ol s(a)$ for every $a>\ol D_{\infty}(\rho\|\sigma)$. Hence,
$\ol s(a)=\phi(a)$ for all $a\in\bR\setminus\{D_{\infty}(\rho\|\sigma)\}$,
from which we obtain
\begin{align}\label{NHsc2}
\scs(r|\rho\|\sigma)=\inf_a\max\left\{r-a,\,\ol s(a) \right\}
=
\inf_a\max\left\{r-a,\,\phi(a) \right\}.
\end{align}
(It is easy to see (e.g., by drawing a picture of the graphs of $\ol s$, $\phi$ and $a\mapsto r-a$) that the values of these functions at $D_{\infty}(\rho\|\sigma)$ do not play a role in the validity of the above identity.) By exactly the same argument as in \cite[Lemma IV.16]{MO}, we have
\begin{align}\label{MOlemma}
\inf_a\max\left\{r-a,\,\ol \phi(a) \right\}=H_r^*(\rho\|\sigma).
\end{align}
This gives an alternative derivation of \eqref{sc for factorization}, based on \eqref{NHsc}--\eqref{MOlemma} and
\eqref{type I limit}--\eqref{type II limit}, and without using Theorems \ref{thm:exponent gen}--\ref{thm:exponent}.
\end{rem}

\section{Examples}
\label{sec:ex}

\subsection{Gibbs states on spin chains}
\label{sec:Gibbs}

Let $\hil$ be a finite-dimensional Hilbert space. A translation-invariant, finite-range \ki{interaction} $\Phi$
on $\hil$
is specified by a number $r\in\bN$, and $\Phi_j\in\B(\hil^{\otimes j}),\,j\in[r]:=\{1,\ldots,r\}$, where each $\Phi_j$ is
self-adjoint.
For every $n\in\bN$, the \ki{local Hamiltonian} $H_n$ corresponding to $\Phi$ is defined as
\begin{align*}
H_n^{\Phi}:=\sum_{j=1}^r\sum_{k:\,k+j-1\le n}\Phi_{j,n,k},\ds\ds\text{where}\\
\Phi_{j,n,k}:=\bz\otimes_{i=1}^{k-1} I\jz\otimes\Phi_j\otimes\bz\otimes_{i=k+j}^{n} I\jz
\end{align*}
is the embedding of $\Phi_j$ into $\B(\hil^{\otimes n})$ from the $k$-th position. The corresponding \ki{local Gibbs state}
on $n$ sites at inverse temperature $\beta>0$ is defined as
\begin{align*}
\omega_n^{\Phi,\beta}:=\frac{e^{-\beta H_n^{\Phi}}}{\Tr e^{-\beta H_n^{\Phi}}}.
\end{align*}
The \ki{thermodynamic limit (TDL) Gibbs state} $\bar\omega_n^{\Phi,\beta}$ on $n$ sites is then given by
\begin{align*}
\Tr A\bar\omega_n^{\Phi,\beta} =\lim_{k\to+\infty}\Tr A\omega_{n+k}^{\Phi,\beta},\ds\ds\ds
A\in\B(\hil^{\otimes n}).
\end{align*}
The existence and the uniqueness of the TDL Gibbs state was shown in \cite{Araki1,Araki2}.
The following has been shown in \cite[Lemma 4.2]{HMO}:
\begin{lemma}\label{lemma:Gibbs factprization}
Let $\Phi$ be a translation-invariant, finite-range interaction, and $\omega:=\{\omega_n^{\Phi,\beta}\}_{n\in\bN}$ and
$\bar\omega:=\{\bar\omega_n^{\Phi,\beta}\}_{n\in\bN}$. Then both $\omega$ and $\bar\omega$ satisfy the factorization property.
\end{lemma}
\smallskip

Lemma \ref{lemma:Gibbs factprization} and Theorem \ref{thm:sc rate with factorization} yield immediately the following:

\begin{thm}
Let $\Phi^{(1)}$ and $\Phi^{(2)}$ be translation-invariant, finite-range interactions on a finite-dimensional Hilbert space
$\hil$, and let $\beta_1,\beta_2>0$. Let $\rho=\{\omega_n^{\Phi^{(1)},\beta_1}\}_{n\in\bN}$ or $\rho=\{\bar\omega_n^{\Phi^{(1)},\beta_1}\}_{n\in\bN}$, and
let $\sigma=\{\omega_n^{\Phi^{(2)},\beta_2}\}_{n\in\bN}$ or $\sigma=\{\bar\omega_n^{\Phi^{(2)},\beta_2}\}_{n\in\bN}$. Then
\begin{align*}
\sci(r|\rho\|\sigma)=\scs(r|\rho\|\sigma)=H_r^*(\rho\|\sigma).
\end{align*}
\end{thm}

\subsection{Quasi-free states of a fermionic lattice}
\label{sec:quasifree sc rate}

In this section we consider the hypothesis testing problem for the case where the null-hypothesis is a temperature state of a
non-interacting fermionic lattice system and the alternative hypothesis is a product state. For the basics on fermionic
quasi-free states, see Appendix \ref{sec:quasi-free review}.
\smallskip

Let $\omega_Q$ and $\omega_R$ be translation-invariant quasi-free states of fermions on the lattice $\bZ^{\nu}$,
with symbols $Q,R\in\B(l^2(\bZ^{\nu}))$. Then $Q$ and $R$ are translation-invariant, and hence there exist
measurable functions $q,r:\,[0,2\pi)^{\nu}\to \bR$ such that $Q=F\inv M_q F$ and $R=F\inv M_r F$, where
$M_q$ and $M_r$ denote the corresponding multiplication operators on $L^2([0,2\pi)^{\nu})$, and
$F$ is the Fourier transformation (see Section \ref{sec:Szego}).
To avoid technical complications, we assume that there exists a $c\in(0,1/2)$ such that $c\le q,r\le 1-c$ almost everywhere with respect to the Lebesgue measure, or equivalently, $c I\le Q,R\le (1-c)I$.

The state of the fermions confined to the hypercube $\C_n:=\{\vect{k}:\,k_1,\ldots,k_{\dimen}=0,\ldots,n-1\}$ is again
a quasi-free state, with symbol $Q_n:=P_nQP_n$ or $R_n:=P_nRP_n$, where
$P_n:=\sum_{k_1,\ldots,k_{\dimen}=0}^{n-1}\pr{\egy_{\{\vect{k}\}}}$,
and $\{\egy_{\{\vect{k}\}}\}_{\vect{k}\in\bZ^{\nu}}$ is the standard basis of $l^2(\bZ^{\nu})$.
These states have density operators on the Fock space $\hil_n:=\F(\ran P_n)$, given by
%The density operators of these states are given by
\begin{align}
\omega_{Q_n}&=\det(I-Q_n)\bigoplus_{k=0}^{n^{\nu}}\bigwedge\nolimits^k \what Q_n,\nonumber\\
\omega_{R_n}&=\det(I-R_n)\bigoplus_{k=0}^{n^{\nu}}\bigwedge\nolimits^k \what R_n,\label{DFP formula}
\end{align}
where $\what Q_n:=Q_n/(I-Q_n)$,  $\what R_n:=Q_n/(I-R_n)$.

With a slight abuse of notation, we identify $\omega_Q$ with $\{\omega_{Q_n}\}_{n\in\bN}$ and
$\omega_R$ with $\{\omega_{R_n}\}_{n\in\bN}$.
We consider the hypothesis testing problem with
\begin{align*}
H_0:\ds\omega_Q\ds\ds\ds\text{ vs. }\ds\ds\ds
H_1:\ds\omega_R.
\end{align*}
Note that the Hilbert space corresponding to one single mode $\egy_{\vect{k}}$ is
$\F(\ran\pr{\egy_{\vect{k}}})\cong\bC^2$, and
\begin{align}\label{Fock decomposition}
\F(\ran P_n)\cong\bigotimes_{\vect{k}\in\C_n}\F(\ran\pr{\egy_{\vect{k}}})\cong(\bC^2)^{\otimes n^{\nu}}.
\end{align}
That is, $\hil_n$ is the Hilbert space of $n^{\nu}$ elementary subsystems. Thus, we replace all the $1/n$ scalings in the previous sections with $1/n^{\nu}$. For instance, we define the strong converse exponents as
\begin{align*}
&\sci(r|\omega_Q\|\omega_R)\\
&\ds:=
\inf\left\{\liminf_{n\to+\infty}-\frac{1}{n^{\nu}}\log\Tr\omega_{Q_n} T_n\Bigm|\right.\\
&\ds\ds\ds\ds\ds\ds\ds\left.\limsup_{n\to\infty}\frac{1}{n^{\nu}}\log\Tr\omega_{R_n}T_n \le -r\right\},\\
&\scs(r|\omega_Q\|\omega_R)\\
&\ds:=
\inf\left\{\limsup_{n\to+\infty}-\frac{1}{n^{\nu}}\log\Tr\omega_{Q_n} T_n\Bigm|\right.\\
&\ds\ds\ds\ds\ds\ds\ds\left.\limsup_{n\to\infty}\frac{1}{n^{\nu}}\log\Tr\omega_{R_n}T_n \le -r\right\}.
\end{align*}
It is easy to verify that the results of Section \ref{sec:differentiable} hold true with appropriately modifying all formulas according to this scaling; e.g., the pinched Neyman-Pearson tests have to be defined as
$\what S_n(a):=\{\what\omega_{Q_n}-e^{n^{\nu}a}\what\omega_{R_n}>0\}$, etc.

We start with showing that $\ol\psi(\alpha|\omega_Q\|\omega_R)$ exists as a limit, and it is a differentiable function of
$\alpha$ for every $\alpha>0$.

\begin{thm}\label{differentiability for quasifree states}
Let $\omega_Q$ and $\omega_R$ be quasi-free states of a fermion system on the lattice $\bZ^{\nu}$, with symbols $Q=F\inv M_qF,\,R=F\inv M_r F$, and assume that
%Let $q,r:\,[0,2\pi)^{\nu}\to\bR$ be measurable, and assume that
there exists a constant $c\in(0,1/2)$ such that $c\le q,r\le 1-c$ almost everywhere with respect to the Lebesgue measure.
Then
for every $\alpha>0$ and $\xx=\oldd$ or $\xx=\neww$, we have
\begin{align}
&\ol\psi(\alpha|\omega_Q\|\omega_R)\nonumber\\
&\ds=
\lim_{n\to+\infty}\frac{1}{n^{\nu}}\psi\x(\alpha|\omega_{Q_n}\|\omega_{R_n})\label{quasifree psi1}\\
&\ds=
\frac{1}{(2\pi)^{\dimen}}\int_{[0,2\pi)^{\dimen}}
\log\Big[ q(\vecc{x})^{\alpha}r(\vecc{x})^{1-\alpha}\nonumber\\
&\ds\ds\ds\ds\ds\ds\ds\ds\ds\ds\ds+(1-q(\vecc{x}))^{\alpha}(1-r(\vecc{x}))^{1-\alpha}\Big] \dd \vecc{x}.\label{quasifree psi2}
\end{align}
Moreover, $\ol\psi(.\,|\omega_Q\|\omega_R)$ is differentiable on $(0,+\infty)$, and
\begin{align}
&\frac{d}{d\alpha}\Big\vert_{\alpha=1}\ol\psi(\alpha|\omega_Q\|\omega_R)\nonumber\\
&\ds=
\frac{1}{(2\pi)^{\dimen}}\int_{[0,2\pi)^{\dimen}}\Bigg[q(\vecc{x})\log\frac{q(\vecc{x})}{r(\vecc{x})}\nonumber\\
&\ds\ds\ds\ds\ds\ds\ds\ds\ds\ds\ds\ds+(1-q(\vecc{x}))\log\frac{1-q(\vecc{x})}{1-r(\vecc{x})}\Bigg] \dd \vecc{x}
\label{quasifree psi3}\\
&\ds=
\lim_{n\to+\infty}\frac{1}{n^{\nu}}D_1\bz\omega_{Q_n}\|\omega_{R_n}\jz.\label{quasifree psi4}
\end{align}
\end{thm}
\begin{proof}
The identities in \eqref{quasifree psi1} and \eqref{quasifree psi2} for $\xx=\oldd$, and the identities in
\eqref{quasifree psi3}--\eqref{quasifree psi4} have been shown in
\cite[Proposition 4.1]{MHOF}. Here we use a similar proof for \eqref{quasifree psi1} and \eqref{quasifree psi2} in the case
$\xx=\neww$.

For the rest, we fix an $\alpha>0$. Let
\begin{align*}
&W_{n,\alpha}:=\bz\frac{Q_n}{1-Q_n}\jz^{\half}\bz\frac{R_n}{1-R_n}\jz^{\frac{1-\alpha}{\alpha}}\bz\frac{Q_n}{1-Q_n}\jz^{\half},\\
&w_{\alpha}:=\bz\frac{q}{1-q}\jz^{\half}\bz\frac{r}{1-r}\jz^{\frac{1-\alpha}{\alpha}}\bz\frac{q}{1-q}\jz^{\half}.
\end{align*}
Then
\begin{align}
&\frac{1}{n^{\nu}}\psi\nw(\alpha|\omega_{Q_n}\|\omega_{R_n})\nn\\
&\ds=
\frac{1}{n^{\nu}}\log\Tr\bz\omega_{Q_n}^{\half}\omega_{R_n}^{\frac{1-\alpha}{\alpha}}\omega_{Q_n}^{\half}\jz^{\alpha}\nn\\
&\ds=
\frac{1}{n^{\nu}}\log\Tr\Bigg(\left[\det(I-Q_n)\F(\what Q_n)\right]^{\half}\left[\det(I-R_n)\F(\what R_n)\right]^{\frac{1-\alpha}{\alpha}}\nn\\
&\ds\ds\ds\ds\ds\ds\ds\ds\ds\ds\ds\left[\det(I-Q_n)\F(\what Q_n)\right]^{\half}\Bigg)^{\alpha}\nn\\
&\ds=
\frac{1}{n^{\nu}}\log\det(I-Q_n)^{\alpha}+\frac{1}{n^{\nu}}\log\det(I-R_n)^{1-\alpha}\nn\\
&\ds\ds\ds+\frac{1}{n^{\nu}}\log\Tr\F(W_{n,\alpha}^{\alpha})\nn\\
&\ds=
\frac{1}{n^{\nu}}\Tr\log (I-Q_n)^{\alpha}+\frac{1}{n^{\nu}}\Tr\log(I-R_n)^{1-\alpha}\nn\\
&\ds\ds\ds+\frac{1}{n^{\nu}}\log\det\bz I+W_{n,\alpha}^{\alpha}\jz\nn\\
&\ds=
\frac{1}{n^{\nu}}\Tr\log (I-Q_n)^{\alpha}+\frac{1}{n^{\nu}}\Tr\log(I-R_n)^{1-\alpha}\nn\\
&\ds\ds\ds+\frac{1}{n^{\nu}}\Tr\log\bz I+W_{n,\alpha}^{\alpha}\jz
\label{quasi-free psi}
\end{align}
By lemma \ref{lemma:Szego}, we have
\begin{align}
&\lim_{n\to+\infty}\left[\frac{1}{n^{\nu}}\Tr\log (I-Q_n)^{\alpha}+\frac{1}{n^{\nu}}\Tr\log(I-R_n)^{1-\alpha}\right]\nn\\
&\ds=
\frac{1}{(2\pi)^{\dimen}}\int_{[0,2\pi)^{\dimen}}\log\left[ (1-q(\vecc{x}))^{\alpha}(1-r(\vecc{x}))^{1-\alpha}\right] \dd \vecc{x}.
\label{quasifree limit1}
\end{align}
To evaluate the limit of the last term in \eqref{quasi-free psi}, we use Corollary \ref{cor:Szego} with
$a^{(1)}=q,\,a^{(2)}=r$, $f^{(1)}(t)=(t/(1-t))^{\half},\, f^{(2)}(t)=(t/(1-t))^{\frac{1-\alpha}{2\alpha}}$ and $g(x)=\log(1+x^{\alpha})$, and
obtain
\begin{align}
&\lim_{n\to+\infty}\frac{1}{n^{\nu}}\Tr\log\bz I+W_{n,\alpha}^{\alpha}\jz\nn\\
&\ds=
\frac{1}{(2\pi)^{\dimen}}\int_{[0,2\pi)^{\dimen}}\log\bz 1+\bz \what q(\vecc{x})^{\half}\what r(\vecc{x})^{\frac{1-\alpha}{\alpha}}\what q(\vecc{x})^{\half}\jz^{\alpha}\jz,
\label{quasifree limit2}
\end{align}
where $\what q:=q/(1-q),\,\what r:=r/(1-r)$.
Combining \eqref{quasi-free psi}, \eqref{quasifree limit1}, and \eqref{quasifree limit2}, we get \eqref{quasifree psi1}--
\eqref{quasifree psi2}. Differentiability of $\ol \psi(\alpha|\omega_Q\|\omega_R)$ is straightforward to verify.
\end{proof}

In particular, Theorem \ref{differentiability for quasifree states} shows that $D_{\alpha}\old$ and $D_{\alpha}\nw$ give rise to the same asymptotic quantities:
\begin{cor}
In the setting of Theorem \ref{differentiability for quasifree states}, we have
\begin{align}\label{quasifree mean Renyi}
&\ol D_{\alpha}\x(\omega_Q\|\omega_R)\\
&\ds:=
\lim_{n\to+\infty}\frac{1}{n^{\nu}}D_{\alpha}\x(\omega_{Q_n}\|\omega_{R_n})\nn\\
&\ds=
\frac{1}{(2\pi)^{\dimen}}\int_{[0,2\pi)^{\dimen}}
\frac{1}{\alpha-1}\log\Big[ q(\vecc{x})^{\alpha}r(\vecc{x})^{1-\alpha}\nn\\
&\ds\ds\ds\ds\ds\ds\ds\ds\ds\ds+(1-q(\vecc{x}))^{\alpha}(1-r(\vecc{x}))^{1-\alpha}\Big] \dd \vecc{x}
\end{align}
for every $\alpha\in(0,+\infty)\setminus\{1\}$ and $\xx=\oldd$ or $\xx=\neww$.
\end{cor}
\smallskip

Now we can obtain the strong converse exponent for quasi-free states.

\begin{thm}\label{thm:quasifree sc}
Let $\omega_Q$ and $\omega_R$ be quasi-free states of a fermion system on the lattice $\bZ^{\nu}$,
and assume that $cI\le Q,R\le (1-c)I$ for some $c\in(0,1/2)$.
Then
\begin{align}\label{quasifree sc rate}
\sci(r|\omega_Q\|\omega_R)&=\scs(r|\omega_Q\|\omega_R)\nn\\
&=H_r^*(\omega_Q\|\omega_R)\nn\\
&=
\sup_{\alpha> 1}\frac{\alpha-1}{\alpha}\left[ r-\ol D_{\alpha}(\omega_Q\|\omega_R)\right],
\end{align}
where $\ol D_{\alpha}(\omega_Q\|\omega_R)$ is given in \eqref{quasifree mean Renyi}.
\end{thm}
\begin{proof}
By Theorem \ref{differentiability for quasifree states}, $\ol\psi(\alpha|\omega_Q\|\omega_R)$ exists as a limit for $\alpha>1$, and the limit is differentiable. By \eqref{DFP formula} and the assumption that $cI\le R\le (1-c)I$,
\begin{align*}
\omega_{R_n}&\ge\det(cI_n)\bigoplus_{k=0}^{n^{\nu}}\bz\frac{c}{1-c}\jz^k I_{\wedge^k\hil_n}\\
&=
c^{n^{\nu}}\bigoplus_{k=0}^{n^{\nu}}\bz\frac{c}{1-c}\jz^k I_{\wedge^k\hil_n}\\
&\ge
\bz c\min\left\{1,\frac{c}{1-c}\right\}\jz^{n^{\nu}}I_{\F(\ran P_n)}.
\end{align*}
Hence, by Corollary \ref{cor:TH2}, we have $\ol\psi(\alpha|\omega_Q\|\omega_R)=\what\psi(\alpha|\omega_Q\|\omega_R)$ for the states
$\what\sigma_n=\what\omega_{R_n}$ in Example \ref{ex:TH}. Combining these two facts, the assertion follows from Theorem \ref{thm:sc rate with differentiability}.
\end{proof}

\appendices

\section{Classical large deviations}
\label{sec:ld}
Let $\mu_n,\,n\in\bN$, be a sequence of finite positive measures on $\bR$, and let $c_n,\,n\in\bN$, be a sequence of positive numbers such that $\lim_n c_n=+\infty$. For each $n$, define the \ki{logarithmic moment generating function} $\lm_n$ by
\begin{equation*}
\lm_n(t):=\log\int_{\bR}e^{tx}\,d\mu_n(x).
\end{equation*}
Here we use the convention $\log +\infty:=+\infty$.
Define
\begin{equation}\label{psibar}
\ol \lm(t):=\limsup_{n\to+\infty}\frac{1}{c_n}\lm_n(c_n t),\ds\ds\ds t\in\bR.
\end{equation}
H\"older's inequality yields that $\lm_n$ is convex for every $n\in\bN$, and hence $\alm$ is convex as well.

The following lemma is a standard generalization of the Markov inequality:
\begin{lemma}\label{lemma:ldp upper bounds}
For every $x\in\bR$,
\begin{align}
\limsup_{n}\frac{1}{c_n}\log\mu_n\bz [x,+\infty)\jz&\le
-\sup_{t\ge 0}\{tx-\overline\lm(t)\},\label{upper1}\\
\limsup_{n}\frac{1}{c_n}\log\mu_n\bz (-\infty,x]\jz&\le -\sup_{t\le 0}\{tx-\overline\lm(t)\}\label{upper2}.
\end{align}
\end{lemma}
\begin{proof}
For every $t\ge 0$,
\begin{align*}
\mu_n\bz [x,+\infty)\jz&=\int_x^{+\infty}\egy_{[x,+\infty)}(z)\,d\mu_n(z)\\
&\le
\int_x^{+\infty}e^{c_nt(z-x)}\,d\mu_n(z)\\
&\le
e^{-c_ntx}\int_{\bR}e^{c_ntz}\,d\mu_n(z),
\end{align*}
and hence,
\begin{align*}
&\limsup_n\frac{1}{c_n}\log\mu_n\bz [x,+\infty)\jz\\
&\ds\le
-tx+\limsup_n\frac{1}{c_n}\log\int_{\bR}e^{c_ntz}\,d\mu_n(z)\\
&\ds= -tx+\alm(t),
\end{align*}
from which \eqref{upper1} follows. The proof of \eqref{upper2} goes the same way.
\end{proof}

The following converse to Lemma \ref{lemma:ldp upper bounds} was essentially given in \cite{DC}, under the stronger
(for our purposes too strong) condition
that $\alm$ exists as a limit in a neighbourhood of $0$, where it is also differentiable.
The more general version below can be easily obtained by following the same line of argument as in
 \cite{DC}. Using the same approach, a generalization of the lower bound in \cite{DC} has
 been obtained in \cite{Chen}, that also generalizes Lemma \ref{lemma:ldp lower bounds} below.
For readers' convenience, we
include a detailed proof below, based on the proof of the G\"artner-Ellis theorem in \cite[pp.~49--50]{DZ}.

\begin{lemma}\label{lemma:ldp lower bounds}
Assume that
$\alm(t)=\lim_n\frac{1}{n}\lm_n(c_n t)$
in some interval $(\alpha,\beta)$, and, moreover, that $\alm$ is a finite-valued differentiable function on $(\alpha,\beta)$.
Then, for every $x\in J:=\bz\lim_{t\searrow \alpha}\alm'(t),\lim_{t\nearrow \beta}\alm'(t)\jz$,
there exists a $t_x\in(\alpha,\beta)$ such that $\alm'(t_x)=x$, and
\begin{align}
\alm^*(x)&:=\sup_{t\in\bR}\{xt-\alm(t)\}=
\sup_{t\in\I}\{xt-\alm(t)\}\nn\\
&=
xt_x-\alm(t_x)>xt-\alm(t),\label{L-transform}
\end{align}
where $\I\subset\bR$ is any interval such that $t_x\in\I$, and the last inequality holds for every
$t\in\bR$ such that $t\notin(\alpha,\beta)$ or $\alm'(t)\ne x$.
 Moreover, for every $x_0,x_1\in \bR\cup\{\pm\infty\}$ such that $x_0<x<x_1$, we have
\begin{align}
\liminf_n\frac{1}{c_n}\log\mu_n\bz (x,x_1)\jz&\ge -\alm^*(x),\label{ldp lowerbound1}\\
\liminf_n\frac{1}{c_n}\log\mu_n\bz (x_0,x)\jz&\ge -\alm^*(x).\label{ldp lowerbound2}
\end{align}
\end{lemma}
\begin{proof}
Since $\alm$ is convex, differentiability on $(\alpha,\beta)$ implies that $\alm'$ is monotone increasing and continuous on
$(\alpha,\beta)$, and hence for every $x\in J$ there exists a $t_x\in(\alpha,\beta)$ such that $\alm'(t_x)=x$.
The rest of the assertions in \eqref{L-transform} are immediate from the concavity of
$t\mapsto xt-\alm(t)$ on $\bR$. Hence, we are left to prove  \eqref{ldp lowerbound1} and  \eqref{ldp lowerbound2}, of which we only prove \eqref{ldp lowerbound1}, as the proof of \eqref{ldp lowerbound2} goes exactly the same way.

Let $x\in J$ and $x_1\in\bR\cup\{+\infty\}$ be such that $x<x_1$.
For every $\delta>0$ such that $x+\delta<x_1,\,(x,x+\delta)\subset J$,
choose a $y\in (x,x+\delta)$. By the above, there is a $t_y\in(\alpha,\beta)$ corresponding to $y$ such that $\alm'(t_y)=y$, and hence, $\alm^*(y)=yt_y-\alm(t_y)$. Since $\alm(t_y)<+\infty$, we have $\lm_n(c_nt_y)<+\infty$ for all large enough $n$, and hence we can define the probability measures $\mu_{n,y}$ by
\begin{align*}
\mu_{n,y}(B)&:=\frac{1}{\int_{\bR}e^{c_nt_ys}\,d\mu_n(s)}\int_B e^{c_nt_ys}\,d\mu_n(s)\\
&=
\int_B e^{c_nt_ys-\lm_n(c_nt_y)}\,d\mu_n(s),
\end{align*}
where $B\subset\bR$ is any Borel set.
Note that
\begin{align*}
&\mu_{n,y}\bz (x,x+\delta)\jz=\int_{(x,x+\delta)}e^{c_nt_ys-\lm_n(c_nt_y)}\,d\mu_n(s)\\
&\ds=
e^{c_nt_yx-\lm_n(c_nt_y)}\int_{(x,x+\delta)}e^{c_nt_y(s-x)}\,d\mu_n(s)\\
&\ds\le
e^{c_nt_yx-\lm_n(c_nt_y)}e^{c_n|t_y|\delta}\int_{(x,x+\delta)}\,d\mu_n(s)\\
&\ds=
e^{c_nt_yx-\lm_n(c_nt_y)}e^{c_n|t_y|\delta}\mu_n\bz (x,x+\delta)\jz,
\end{align*}
and therefore,
\begin{align}
&\liminf_n\frac{1}{c_n}\log\mu_n\bz (x,x_1)\jz\nn\\
&\ds\ge
\liminf_n\frac{1}{c_n}\log\mu_n\bz (x,x+\delta)\jz\nonumber\\
&\ds\ge
\alm(t_y)-t_yx-|t_y|\delta+\liminf_n\frac{1}{c_n}\log\mu_{n,y}\bz (x,x+\delta)\jz,\label{ldp lower1}
\end{align}
where we used that $\alm(t_y)$ exists as a limit.
If we can prove that
\begin{align}
0>\max\Bigg\{&\limsup_{n\to+\infty}\frac{1}{c_n}\log\mu_{n,y}\bz (-\infty,x]\jz,\nn\\
&\limsup_{n\to+\infty}\frac{1}{c_n}\log\mu_{n,y}\bz [x+\delta,+\infty)\jz
\Bigg\}\label{ldp lower2}
\end{align}
then we have $\lim_{n\to+\infty}\mu_{n,y}\bz (x,x+\delta)\jz=1$, and hence, by \eqref{ldp lower1},
\begin{align*}
\liminf_n\frac{1}{c_n}\log\mu_n\bz (x,x_1)\jz&\ge
\alm(t_y)-t_yx-|t_y|\delta\\
&> -\alm^*(x)-|t_y|\delta,
\end{align*}
where the second inequality is due to \eqref{L-transform}.
Using that $|t_y|\delta\to 0$ as $\delta\searrow 0$, \eqref{ldp lowerbound1} follows.
Hence, we are left to prove \eqref{ldp lower2}.

Let $\lm_{n,y}$ denote the logarithmic moment generating function of $\mu_{n,y}$, i.e., for every $t\in\bR$,
\begin{align*}
\lm_{n,y}(c_nt)&:=\log\int_{\bR}e^{c_nts}\,d\mu_{n,y}(s)\\
&=
\log\int_{\bR}e^{c_nts+c_nt_ys-\lm_n(c_nt_y)}\,d\mu_n(s)\\
&=\lm_n(c_nt+c_nt_y)-\lm_n(c_nt_y),
\end{align*}
and let $\alm_y(t):=\limsup_{n\to+\infty}\frac{1}{c_n}\lm_{n,y}(c_nt)$. By assumption,
\begin{align}\label{ldp lower3}
\alm_y(t)=\alm(t+t_y)-\alm(t_y)
\end{align}
for all $t$ such that $t+t_y\in (\alpha,\beta)$.
By lemma \ref{lemma:ldp upper bounds},
\begin{align*}
\limsup_{n\to+\infty}\frac{1}{c_n}\log\mu_{n,y}\bz [x+\delta,+\infty)\jz&\le-\sup_{t\ge 0}\{t(x+\delta)-\alm_y(t)\},\\
\limsup_{n\to+\infty}\frac{1}{c_n}\log\mu_{n,y}\bz (-\infty,x]\jz&\le-\sup_{t\le 0}\{tx-\alm_y(t)\},
\end{align*}
and hence \eqref{ldp lower2} will be proved if we can show that
\begin{equation}\label{formula2}
0<\sup_{t\ge 0}\{t(x+\delta)-\alm_y(t)\}\ds\ds\text{and}\ds\ds
0<\sup_{t\le 0}\{tx-\alm_y(t)\}.
\end{equation}
By \eqref{ldp lower3},
\begin{align*}
&\sup_{t\ge 0}\{t(x+\delta)-\alm_y(t)\}\\
&\ds=
\sup_{t\ge 0}\{t(x+\delta)-\alm(t+t_y)+\alm(t_y)\}\nonumber\\
&\ds=
\alm(t_y)-(x+\delta)t_y+\sup_{t\ge 0}\{(t+t_y)(x+\delta)-\alm(t+t_y)\}\nonumber\\
&\ds=
\alm(t_y)-(x+\delta)t_y+\sup_{t\ge t_y}\{t(x+\delta)-\alm(t)\}\nonumber\\
&\ds=
\alm^*(x+\delta)-\{(x+\delta)t_y-\alm(t_y)\}\\
&\ds>0,
\end{align*}
where the last identity and the inequality follows from
\eqref{L-transform}, since $\alm'(t_y)=y< x+\delta$.
The other half of \eqref{formula2} follows by the same kind of argument, which we omit.
\end{proof}

\section{Fermionic quasi-free states}
\label{sec:quasi-free review}

For a separable Hilbert space $\hil$ and $k\in\bN$, let $\wedge^k\hil$ denote the $k$-th antisymmetric tensor power of $\hil$, with the convention
$\wedge^0\hil:=\bC$. Given $x_1,\ldots,x_k\in\hil$, their anti-symmetrized tensor product is defined as
\begin{equation*}
x_1\wedge\ldots\wedge x_k:=\frac{1}{\sqrt{n!}}\sum_{\sigma\in S_k} s(\sigma) x_{\sigma(1)}\otimes\ldots\otimes x_{\sigma(k)},
\end{equation*}
where the sum runs over all permutations of $k$ points. We have $\wedge^k\hil=\spann\{x_1\wedge\ldots\wedge x_k:\,x_i\in\hil\}$.
The \ki{anti-symmetric-} or \ki{fermionic Fock space} $\F(\hil)$ is defined as
\begin{align*}
\F(\hil):=\bigoplus_{k=0}^{\dim\hil}\bigwedge\nolimits^k\hil,
\end{align*}
where $\dim\hil$ may be countably infinite. Note that $\dim\F(\hil)=2^{\dim\hil}$ when $\dim\hil<+\infty$, and otherwise
$\F(\hil)$ is countably infinite-dimensional.
In the physics terminology, $\F(\hil)$ is the Hilbert space of a system of at most $\dim\hil$ fermions,
and the pure state $\pr{x_1\wedge\ldots\wedge x_k}$ describes $k$ fermions in the modes $x_1,\ldots,x_k$.

For each $x\in\hil$, the corresponding \ki{creation operator} is defined as the unique bounded linear extension $ c^*(x):\,\F(\hil)\to\F(\hil)$ of
\begin{equation*}
c^*(x):\,x_1\wedge\ldots\wedge x_k\mapsto x\wedge x_1\wedge\ldots\wedge x_k\,
\end{equation*}
where $x_1,\ldots,x_k\in\hil,\, k\in\bN$,
and the corresponding \ki{annihilation operator} is its adjoint $c(x):=\bz c^*(x)\jz^*$.
The interpretation is that $c^*(x)$ creates a fermion in the mode $x$.
Creation and annihilation operators satisfy the \ki{canonical anticommutation relations (CAR)}:
$c(x)c(y)+c(y)c(x)=0$ and $c(x)c^*(y)+c^*(y)c(x)=\inner{x}{y}\unit$ for every $x,y\in\hil$.

Observable quantities of the system are elements of the algebra
$\A(\hil)$ generated by the creation and the annihilation operators. When $\hil$ is finite-dimensional,
$\A(\hil)$ is equal to all the bounded operators on $\F(\hil)$. In the infinite-dimensional case, we need to take the closure in
some topology; closure
in the norm topology yields the so-called \ki{CAR algebra} $\car(\hil)$, which is strictly smaller than the closure in the weak
topology, which is $\B(\F(\hil))$.

Note that for any $A\in\B(\hil)$,
$A^{\otimes k}$ leaves $\wedge^k\hil$ invariant, and we denote the restriction of $A^{\otimes k}$ onto $\wedge^k\hil$ by $\wedge^k A$. If $\norm{A}\le 1$ or $\dim\hil<+\infty$ then
\begin{align*}
\F(A):=\bigoplus_{k=0}^{\dim\hil}\bigwedge\nolimits^k A
\end{align*}
is a bounded operator on $\F(\hil)$.
If $\hil$ is finite-dimensional and $A$ has eigenvalues $\lambda_1,\ldots,\lambda_d$, counted with multiplicities, then the eigenvalues of $\wedge^k A$ are
$\{\lambda_{i_1}\cdot\ldots\cdot\lambda_{i_k}\,:\,i_1<\ldots<i_k\}$. Thus we get that in this case
\begin{equation*}
\Tr\F(A)=\det(I+A).
\end{equation*}

Given an operator $Q\in\B(\hil)_+$ such that $Q\le I$, there exists a unique positive linear functional $\omega_Q$ on $\car(\hil)$ such that
$\omega_Q(I)=1$, and for any $x_1,\ldots,x_n,y_1,\ldots,y_m$,
\begin{align*}
&\omega_{Q}\,\bz c(x_1)^*\ldots c(x_n)^* c(y_m)\ldots c(y_1)\jz\\
&\ds= \delta_{m,n}\det \{\inner{y_i}{Q\s x_j}\}_{i,j=1}^n.
\end{align*}
That is, $\omega_Q$ is uniquely determined by its two-point correlation functions on creation and annihiliation operators.
Such a functional $\omega_Q$ is called a \ki{(gauge-invariant) quasi-free state}, and $Q$ the \ki{symbol} of the state.
If $\hil$ is finite-dimensional then $\omega_Q$ can be given by a density operator on $\F(\hil)$ which, with a slight abuse of notation, we also denote by $\omega_Q$.
If, moreover, $Q<I$ then $\omega_Q$ can be written explicitly as
\begin{equation*}
\omega_{Q}=\det(I-Q)\bigoplus_{k=0}^{\dim\hil}\bigwedge\nolimits^k\frac{Q}{I-Q}=\det(I-Q)\F\bz\frac{Q}{I-Q}\jz,
\end{equation*}
according to \cite[Lemma 3]{DFP}.

The dynamics of a system of non-interacting fermions is determined by a single-particle Hamilton operator, i.e., a self-adjoint operator $H$ on $\hil$. Assume for the rest that $\hil$ is finite-dimensional, and
let $H_{k,i}:=I^{\otimes (i-1)}\otimes H\otimes I^{\otimes (k-i)}$ be the embedding of $H$ into the $i$-th position in $\hil^{\otimes k}$. It is easy to see that
$\sum_{i=1}^kH_{k,i}$ leaves $\wedge^k\hil$ invariant, and we denote the restriction of  $H_{k,i}$ onto
$\wedge^k\hil$ by $\Gamma_k(H)$, and define the second-quantized Hamiltonian $\Gamma(H):=\oplus_{k=1}^{\dim\hil}\Gamma_k(\hil)$.
If the initial state of the system is a pure state given by the vector
$x_1\wedge\ldots\wedge x_r$ then the state after time $t$ is the pure state given by the vector
%\begin{align*}
$\bz e^{-itH}x_1\jz\wedge\ldots\wedge\bz e^{-itH}x_r\jz=
e^{-it\Gamma(H)}(x_1\wedge\ldots\wedge x_r)$.
Thus, the dynamics of the many-particle system is governed by the Hamiltonian $\Gamma(H)$. Hence, the equilibrium state at
inverse temperature $\beta$ (Gibbs state) is $e^{-\beta\Gamma(H)}/\Tr e^{-\beta\Gamma(H)}$, and a direct computation shows that
this is a quasi-free state with
\begin{align*}
Q=\frac{e^{-\beta H}}{I+e^{-\beta H}}.
\end{align*}
Vice versa, any quasi-free state with symbol $Q$ such that $0<Q<1$ is the Gibbs state of non-interacting fermions at some
inverse temperature $\beta$. The infinite-dimensional case is slightly more complicated: it is still true that the equilibrium
state of non-interacting fermions at inverse temperature $\beta$ with one-particle Hamiltonian $H$ is the quasi-free state with
symbol $e^{-\beta H}/(I+e^{-\beta H})$, but the equilibrium state in this case is defined through the KMS condition
\cite[Section 5.2.4]{BR2}.

To describe a system of non-interacting fermions occupying sites of a $\nu$-dimensional cubic lattice, we choose
$\hil:=l^2(\bZ^{\nu})$.
As we have seen, the equilibrium state of the system at any inverse temperature is a quasi-free state, with some symbol
$Q\in\B(l^2(\bZ^{\nu}))$.
Let $\{\egy_{\{\vect{k}\}}\,:\,\vect{k}\in\Z^{\dimen}\}$ denote
the standard basis of $l^2(\Z^{\dimen})$, and for every
$n\in\bN$, let $P_n:=\sum_{k_1,\ldots,k_{\dimen}=0}^{n-1}\pr{\egy_{\{\vect{k}\}}}$. The Hilbert space of the fermions occupying
sites of the hypercube $\C_n:=\{\vect{k}:\,k_1,\ldots,k_{\dimen}=0,\ldots,n-1\}$ is then
$\F(P_n\hil)$, and the state of this subsystem is quasi-free with symbol $Q_n:=P_nQP_n$.

The \ki{translation operators} are the unique linear extensions of
$S_{\vect{j}}\egy_{\{\vect{k}\}}\mapsto \egy_{\{\vect{k}+\vect{j}\}},\s
\vect{k},\vect{j}\in\Z^{\dimen}$.
The map $\gamma_{\vect{j}}(c(x)):=c\bz S_\vect{j}x\jz$ extends to an automorphism of CAR$\bz l^2(\Z^{\dimen})\jz$ for all
$\vect{j}\in\Z^{\dimen}$,
and $\gamma_{\vect{j}},\s \vect{j}\in\Z^{\dimen}$, is a group of automorphisms, called the group of \ki{translation
automorphisms}. A quasi-free state $\omega_Q$ is called \ki{translation-invariant} if $\omega_{Q}\circ\gamma_\vect{j}=\omega_Q,\s \vect{j}\in\Z^{\dimen}$, which
holds if and only if its symbol $Q$ is translation-invariant, i.e., it commutes with all the unitaries $S_{\vect{j}},\s  \vect{j}\in\Z^{\dimen}$.

\section{Generalizations of Szeg\H o's theorem}
\label{sec:Szego}

Translation-invariant operators on $l^2(\Z^{\dimen})$ commute with each other, and they are simultaneously diagonalized by the Fourier transformation
\begin{align*}
&F:\,l^2(\Z^{\dimen})\to L^2([0,2\pi)^{\dimen})\\
&F\egy_{\{\vect{k}\}}:=\vfi_{\vect{k}}\,,\ds \vfi_{\vect{k}}(\vecc{x}):=e^{i\inner{\vect{k}}{\vecc{x}}}\,,\ds\vecc{x}\in[0,2\pi)^{\dimen}\,,\vect{k}\in\Z^{\dimen}\,,
\end{align*}
where $\inner{\vect{k}}{\vecc{x}}:=\sum_{i=1}^{\dimen}k_ix_i$. That is,
every translation-invariant operator $A$ arises in the form $A=F^{-1}M_{ a}F$, where $M_{ a}$ denotes the multiplication operator by a bounded measurable function $a$ on $[0,2\pi)^{\dimen}$.
In the rest of this section we use the convention that lower and upper case versions of the same letter refer to a measurable function on $[0,2\pi)^{\dimen}$ and the Fourier transform of the corresponding multiplication operator, respectively. Moreover, we use the notation
\begin{align*}
&A_n:=P_nAP_n,\ds\ds\ds A\in\B\bz l^2(\Z^{\dimen})\jz, \ds\ds\ds\text{where}\\
&P_n:=\sum_{k_1,\ldots,k_{\dimen}=0}^{n-1}\pr{\egy_{\{\vect{k}\}}}.
\end{align*}

Let $\Sigma(A)$ denote the convex hull of the spectrum of a self-adjoint operator $A$.
The following generalization of Szeg\H o's theorem \cite{GSz} has been shown in \cite[lemma 3.1]{MHOF}:
\begin{lemma}\label{lemma:Szego}
For all $k=1\ldots,r$, let $ a^{(k)}:\,[0,2\pi)^{\dimen}\to\bR$ be a bounded measurable function, and
$f^{(k)}:\,\Sigma(A^{(k)})\to\bR$ be continuous. Then
\begin{align}\label{convergence}
&\lim_{n\to \infty}\frac{1}{n^{\dimen}}\Tr f^{(1)}( A_n^{(1)})\cdot\ldots\cdot f^{(r)}( A_n^{(r)})\nn\\
&\ds=
\frac{1}{(2\pi)^{\dimen}}\int_{[0,2\pi)^{\dimen}} f^{(1)}(  a^{(1)}(\vecc{x}))\cdot\ldots\cdot f^{(r)}( a^{(r)}(\vecc{x})) \dd \vecc{x}.
\end{align}
\end{lemma}

From this lemma, we can easily get the following:

\begin{cor}\label{cor:Szego}
In the setting of lemma \ref{lemma:Szego}, we have
\begin{align*}
&\lim_{n\to \infty}\frac{1}{n^{\dimen}}\Tr g\bz \left[\prod_{k=1}^r f^{(k)}( A_n^{(k)})\right]\left[\prod_{k=1}^r f^{(k)}( A_n^{(k)})\right]^*\jz\\
%&\ds\ds\ds\ds\ds\ds\ds\ds\ds\ds\ds\ds\ds\ds\ds\ds=
&\ds=\frac{1}{(2\pi)^{\dimen}}\int_{[0,2\pi)^{\dimen}} g\bz \left[\prod_{k=1}^r f^{(k)}( a^{(k)})\right]^2\jz\dd \vecc{x}
\end{align*}
for any continuous function $g:\,[0,\Delta]\to\bR$, where $\Delta:=\prod_{k=1}^r\max_{s\in\Sigma(A^{(k)})}|f^{(k)}(s)|^2$.
\end{cor}
\begin{proof}
Let $B_n:=\prod_{k=1}^r f^{(k)}( A_n^{(k)})$ and $b:=\prod_{k=1}^r f^{(k)}( a^{(k)})$.
Since $g$ is continuous on $[0,\Delta]$, the Stone-Weierstrass theorem
tells that for every $\ep>0$, there exists a polynomial $g_{\ep}$ such that
\begin{align*}
\norm{g-g_{\ep}}_{\infty}:=\max_{x\in[0,\Delta]}|g(x)-g_{\ep}(x)|<\ep.
\end{align*}
Now, for a fixed $\ep>0$, lemma \ref{lemma:Szego} yields that there exists an $N_{\ep}$ such that for all $n\ge N_{\ep}$,
\begin{align*}
\abs{\frac{1}{n^{\dimen}}\Tr g_{\ep}\bz B_nB_n^*\jz
-
\frac{1}{(2\pi)^{\dimen}}\int_{[0,2\pi)^{\dimen}} g_{\ep}\bz b(\vecc{x})^2\jz\dd \vecc{x}}<\ep.
\end{align*}
Hence, for every $n\ge N_{\ep}$,
\begin{align*}
&\abs{\frac{1}{n^{\dimen}}\Tr g\bz B_nB_n^*\jz
-\frac{1}{(2\pi)^{\dimen}}\int_{[0,2\pi)^{\dimen}} g\bz b(\vecc{x})^2\jz\dd \vecc{x}}\\
&\ds\le
\abs{\frac{1}{n^{\dimen}}\Tr g\bz B_nB_n^*\jz-\frac{1}{n^{\dimen}}\Tr g_{\ep}\bz B_nB_n^*\jz
}\\
&\ds\ds+\abs{\frac{1}{n^{\dimen}}\Tr g_{\ep}\bz B_nB_n^*\jz
-
\frac{1}{(2\pi)^{\dimen}}\int_{[0,2\pi)^{\dimen}} g_{\ep}\bz b(\vecc{x})^2\jz\dd \vecc{x}}\\
&\ds\ds
+\abs{\frac{1}{(2\pi)^{\dimen}}\int_{[0,2\pi)^{\dimen}} g_{\ep}\bz b(\vecc{x})^2\jz\dd \vecc{x}
-
\frac{1}{(2\pi)^{\dimen}}\int_{[0,2\pi)^{\dimen}} g\bz b(\vecc{x})^2\jz\dd \vecc{x}}\\
&\ds\le
2\norm{g-g_{\ep}}_{\infty}+\ep
<3\ep.
\end{align*}
\end{proof}

\section{Classical i.i.d. and Markov chains}
\label{sec:classical}

Let $\X$ be a finite set, and
$\rho$ and $\sigma$ be probability measures on the sigma-field generated by the cylinder sets of
$\X^{\infty}:=\times_{k=1}^{\infty}\X$. We denote by $\rho_n$ and $\sigma_n$ the restrictions of
$\rho$ and $\sigma$, respectively, to $\X^n=\times_{k=1}^n\X$. Then $\rho_n$ and $\sigma_n$ can be identified with their respective probability mass functions, which we also denote by $\rho_n$ and $\sigma_n$.
Let $\hil_n:=l^2(\X^n)=l^2(\X)^{\otimes n}$, where for any finite set $Y$,
$l^2(Y)=\bC^Y$ equipped with the inner product $\inner{f}{g}:=\sum_{y\in \Y}\ol f(y)g(y)$. The
multiplication operators by $\rho_n$ and $\sigma_n$ on $l^2(X)^{\otimes n}$ are density operators, which we also denote by $\rho_n$ and $\sigma_n$ if no confusion arises.
Moreover, we identify the probability measures $\rho$ and $\sigma$ with $\{\rho_n\}_{n\in\bN}$ and $\{\sigma_n\}_{n\in\bN}$.

Our aim here is to show how the expressions in \cite{NKiid,NKMarkov} for the strong converse exponent in the classical i.i.d. and Markov case can be recovered from our Theorems
\ref{thm:sc rate with differentiability} and \ref{thm:iid}.
Note that in the classical case any choice of the auxiliary sequence $\what\sigma$
yields $\what\rho_n=\rho_n,\,n\in\bN$, and thus
$\ol\psi=\what\psi$. Hence, in the classical case it is sufficient to verify the differentiability of $\ol\psi$ to apply
Theorem \ref{thm:sc rate with differentiability}.
\medskip

First, we consider the i.i.d.~case, where $\rho_n=\rho_1^{\otimes n}$, $\sigma_n=\sigma_1^{\otimes n}$, $n\in\bN$. We assume that
$\supp\rho_1\subseteq\supp\sigma_1$. Then
%with the choice $\what\sigma_n:=\sigma_n$, we have $\what\rho_n=\rho_n$, and
$\psi(t):=\ol\psi(t|\rho\|\sigma)=\psi(t|\rho_1\|\sigma_1)=\log Z(t)$, where
$Z(t):=\sum_{x}\rho_1(x)^t\sigma_1(x)^{1-t}$. Define
\begin{align*}
\omega_1\t(x):=\rho_1(x)^t\sigma_1(x)^{1-t}/Z(t),\ds\ds\ds x\in\X,\ds t\in\bR,
\end{align*}
and let
\begin{align*}
\omega_1^{(\infty)}(x):=\lim_{t\to+\infty}\omega_1\t(x)=
\begin{cases}
0,&x\notin\X^*,\\
\sigma_1(x)/\sigma(\X^*),&x\in\X^*,
\end{cases}
\end{align*}
where $x\in\X^*$ if $\rho_1(y)/\sigma_1(y)\le\rho_1(x)/\sigma_1(x)$ for all $y\in\X$, i.e.,
if $\log\bz\rho_1(x)/\sigma_1(x)\jz=D_{\max}(\rho_1\|\sigma_1)$.
Then $\omega_1\t$ is a probability mass function on $\X$ for every $t\in\bR\cup\{+\infty\}$, and a straightforward computation shows
\begin{align}
&D(\omega_1\t\|\sigma_1)=-\psi(t)+t\psi'(t),\label{NK3}\\
&D(\omega_1^{(\infty)}\|\sigma_1)=-\log\sigma(\X^*),\\
&D(\omega_1\t\|\rho_1)=-\psi(t)+(t-1)\psi'(t),\\
&D(\omega_1^{(\infty)}\|\rho_1)=-\log\sigma(\X^*)-D_{\max}(\rho_1\|\sigma_1),\label{NK4}\\
&\tau(t):=D(\omega_1\t\|\sigma_1)-D(\omega_1\t\|\rho_1)=\psi'(t),\\
&\tau(\infty):=\lim_{t\to+\infty}\tau(t)=D_{\max}(\rho_1\|\sigma_1).
\end{align}
Using now the notations and results of Section \ref{sec:LF} with $f:=\psi$, we get
\begin{align*}
D_{\psi,1}&=D(\rho_1\|\sigma_1)=D(\omega_1^{(1)}\|\sigma_1),\\
D_{\psi,\infty}&=D_{\max}(\rho_1\|\sigma_1),\\
r_{\max}&=\sup_{1<t<+\infty}(-\psi(t)+t\psi'(t))=\sup_{1<t<+\infty}D(\omega_1\t\|\sigma_1)\\
&=D(\omega_1^{(\infty)}\|\sigma_1),
\end{align*}
where in the last line we used \eqref{rmax} and the fact that $t\mapsto -\psi(t)+t\psi'(t)=F'\bz\frac{t-1}{t}\jz$ is monotone increasing due to the
convexity of $F$. Using now \eqref{diff r}--\eqref{diff Hr} and \eqref{NK3}--\eqref{NK4}, we see that for every
$r\in(D(\omega_1^{(1)},\sigma_1),D(\omega_1^{(\infty)}\|\sigma_1))$,
there exists a $t_r\in(1,+\infty)$ such that
\begin{align}\label{NK1}
r&=-\psi(t_r)+t\psi'(t_r)=D(\omega_1^{(t_r)}\|\sigma),\nn\\
H_r^*(\rho_1\|\sigma_1)&=-\psi(t_r)+(t_r-1)\psi'(t_r)=D(\omega_1^{(t_r)}\|\rho),
\end{align}
and for $r\ge D(\omega_1^{(\infty)}\|\sigma_1)$,
\begin{align}\label{NK2}
H_r^*(\rho_1\|\sigma_1)&=r-D_{\max}(\rho_1\|\sigma_1)\nn\\
&=r-D(\omega_1^{(\infty)}\|\sigma_1)+D(\omega_1^{(\infty)}\|\rho_1),
\end{align}
due to \eqref{Hr expressions}, \eqref{amax} and \eqref{NK3}--\eqref{NK4}. Combining now Theorem \ref{thm:iid} with
\eqref{NK1}--\eqref{NK2}, we get Theorems 2 and 3 in \cite{NKiid} with $p_0=\sigma_1,\,p_1=\rho_1$.
\medskip

Next, we consider the case where $\rho$ and $\sigma$ are Markov chains, with transition matrices $R$ and $S$, respectively.
That is, for any $x_1,\ldots,x_n\in\X$, we have
\begin{align*}
\rho_n(x_1,\ldots,x_n)&=\rho_1(x_1)R_{x_1x_2}R_{x_2x_3}\ldots R_{x_{n-1}x_n},\\
\sigma_n(x_1,\ldots,x_n)&=\sigma_1(x_1)S_{x_1x_2}S_{x_2x_3}\ldots S_{x_{n-1}x_n}.
\end{align*}
We assume that
$R$ is irreducible, i.e., there exists an $n\in\bN$ such that all the entries of $(I+R)^n$ are strictly positive. We also assume that
$\supp R\subseteq\supp S$, i.e., $R_{xy}>0\imp S_{xy}>0$ for all $x,y\in\X$; then $S$ is also irreducible. By the Perron-Frobenius theory (see, e.g., \cite[Theorem 3.1.1]{DZ}), for any irreducible matrix
with non-negative entries, the spectral radius of the matrix is an eigenvalue, and the corresponding left and right eigenvectors can be chosen to have strictly positive entries.
Here we don't assume that the Markov chains are stationary, only that their initial distributions are strictly positive, i.e., $\rho_1(x)>0,\,\sigma_1(x)>0$
for all $x\in\X$. For every $t\in\bR$, let $T_{xy}(t):=R_{xy}^tS_{xy}^{1-t}$. Then $T(t)$ is irreducible for every $t\in\bR$. Let $\lambda(t)$ be the spectral radius of $T(t)$, and let $v(t)$ be a corresponding right eigenvector
of $T(t)$ with strictly positive entries.
Noting that $v_x(t)/M(t)\le 1\le v_x(t)/m(t)$ for every $x\in\X$, where $m(t):=\min_{x}v_x(t),\,M(t):=\max_x v_x(t)$, we get
\begin{align*}
\ol\psi(t)&:=\ol\psi(t|\rho\|\sigma)\\
&=\lim_{n\to+\infty}\frac{1}{n}\log\sum_{x_1,\ldots,x_n\in\X}
\rho_1(x_1)^t\sigma_1(x_1)^{1-t}\\
&\ds\ds\ds\ds\ds\ds\ds\ds\ds\ds\ds\ds\ds  R_{x_1x_2}^tS_{x_1x_2}^{1-t}\ldots R_{x_{n-1}x_n}^tS_{x_{n-1}x_n}^{1-t}\\
&=
\lim_{n\to+\infty}\frac{1}{n}\log\inner{u(t)}{T(t)^{n-1}\mathbf{1}}=\log\lambda(t),
\end{align*}
where $u_x(t):=\rho_1(x)^t\sigma_1(x)^{1-t},\,x\in\X$, and $\mathbf{1}$ stands for the constant one vector
(see also \cite[Theorem 3.1.1]{DZ}). By standard results in perturbation theory \cite{Kato}, $\lambda(t)$ is an analytic function of $t$. In particular, $\ol\psi(t|\rho\|\sigma)$ is a differentiable function of $t$, and Theorem
\ref{thm:sc rate with differentiability} yields that \eqref{differentiable sc rate} holds. Our aim now is derive alternative expressions for
$H_r^*(\rho\|\sigma)$, from which we can recover the results of \cite{NKMarkov}.

Following \cite{Vasek,Natarajan}, we define
\begin{align*}
Q_{x,y}(t):=\frac{T_{xy}(t)v_y(t)}{\lambda(t)v_x(t)}.
\end{align*}
Then $Q(t)$ is irreducible with stationary distribution $q(t)$, and we denote the generated irreducible Markov chain
by $\omega\t$. The following explicit expressions for the asymptotic relative entropies are easy to verify:
\begin{align*}
\ol D(\omega\t\|\sigma)&:=\lim_{n\to+\infty}\frac{1}{n}D(\omega\t_n\|\sigma_n)\\
&=
\sum_{x,y}q_x(t)Q_{xy}(t)\log\frac{Q_{xy}(t)}{S_{xy}(t)},\\
\ol D(\omega\t\|\rho)&:=\lim_{n\to+\infty}\frac{1}{n}D(\omega\t_n\|\rho_n)\\
&=
\sum_{x,y}q_x(t)Q_{xy}(t)\log\frac{Q_{xy}(t)}{R_{xy}(t)}.
\end{align*}
Now we follow a modification of the proof of \cite[Lemma 2.3]{Vasek} to connect the above formulas to $\ol\psi$.
Normalizing $v(t)$ such that $\sum_x q_x(t)v_x(t)=1$, we have, for every $n\in\bN$,
\begin{align*}
\ol\psi(t|\rho\|\sigma)&=\log\lambda(t)\\
&=
\frac{1}{n}\log\sum_{x_0,x_1,\ldots,x_n}\frac{q_{x_0}(t)}{v_{x_0}(t)}R_{x_0x_1}^tS_{x_0x_1}^{1-t}\ldots\\
&\ds\ds\ds\ds\ds\ds\ds\ds\ds\ds\ds\ldots R_{x_{n-1}x_n}^tS_{x_{n-1}x_n}^{1-t}v_{x_n}(t).
\end{align*}
A straightforward calculation gives
\begin{align}
&\ol\psi'(t)\nn\\
&\ds=\frac{1}{t}\log\lambda(t)\label{Vasek0}\\
&\ds\ds+\frac{1}{t}\frac{1}{n}\sum_{x_0,x_1,\ldots,x_n}q_{x_0}(t)Q_{x_0x_1}\ldots Q_{x_{n-1}x_n}\nn\\
&\ds\ds\ds\ds\ds\ds\ds\ds\ds\ds\ds\ds\ds\log\frac{Q_{x_0x_1}\ldots Q_{x_{n-1}x_n}}{S_{x_0x_1}\ldots S_{x_{n-1}x_n}}\label{Vasek1}\\
&\ds\ds+\frac{1}{t}\frac{1}{n}\sum_{x_0,x_1,\ldots,x_n}q_{x_0}(t)Q_{x_0x_1}\ldots Q_{x_{n-1}x_n}\left[v_{x_0}(t)-v_{x_n}(t)\right]
\label{Vasek2}\\
&\ds\ds+\frac{1}{n}\sum_{x_0,x_1,\ldots,x_n}q_{x_0}(t)Q_{x_0x_1}\ldots Q_{x_{n-1}x_n}\frac{d}{dt}\frac{q_{x_0}(t)v_{x_n}(t)}{v_{x_0}(t)}.\label{Vasek3}
\end{align}
The term in \eqref{Vasek2} is equal to zero, and taking the limit $n\to+\infty$ yields
\begin{align}
\ol\psi'(t)=\frac{1}{t}\log\lambda(t)+\frac{1}{t}\ol D(\omega\t\|\sigma).
\end{align}
It is easy to see that \eqref{Vasek0}--\eqref{Vasek3} holds also if \eqref{Vasek0}--\eqref{Vasek1} is replaced with
\begin{align*}
\frac{1}{t-1}\log\lambda(t)+\frac{1}{t-1}\frac{1}{n}\sum_{x_0,x_1,\ldots,x_n}&q_{x_0}(t)Q_{x_0x_1}\ldots Q_{x_{n-1}x_n}\\
&\log\frac{Q_{x_0x_1}\ldots Q_{x_{n-1}x_n}}{R_{x_0x_1}\ldots R_{x_{n-1}x_n}},
\end{align*}
and as above, we obtain
\begin{align}
\ol\psi'(t)=\frac{1}{t-1}\log\lambda(t)+\frac{1}{t-1}\ol D(\omega\t\|\rho).
\end{align}
Hence,
\begin{align}
\ol D(\omega\t\|\sigma)&=-\ol\psi(t)+t\ol\psi'(t),\\
\ol D(\omega\t\|\rho)&=-\ol\psi(t)+(t-1)\ol\psi'(t),
\end{align}
in complete analogy with \eqref{NK3}--\eqref{NK4}.
Using again the general considerations in Section \ref{sec:LF} with $f=\ol\psi$, and Theorem \ref{thm:sc rate with differentiability}, we recover Theorems 2 and 3 from \cite{NKMarkov}.
\medskip

We refer to \cite{NKMarkov,HW,WH} for more details on exponentially decaying tail probabilities, hypothesis testing, and the information geometry of classical Markov chains.

\section*{Acknowledgments}

MM would like to thank Prof.~Fumio Hiai for discussions on the eigenvalues of Gibbs states,
and Vincent F.~Tan and Marco Tomamichel for pointing out the paper \cite{Chen}.
This work was partially supported by the MEXT Grant-in-Aid
(A) No.~20686026 ``Project on Multi-user Quantum Network'' (TO),
and by
the European Research Council Advanced Grant ``IRQUAT'', the
Spanish MINECO  Project No. FIS2013-40627-P, the Generalitat de Catalunya CIRIT Project No. 2014 SGR 966, and
by the Technische Universit\"at M\"unchen -- Institute for Advanced Study, funded by the German Excellence Initiative and the European Union Seventh Framework Programme under grant agreement no. 291763
 (MM). The authors are grateful to two anonymous referees for their comments that helped to improve the paper, and in particular for one referee for suggesting the application of the construction from the poof of \cite[Theorem 14]{TH}.

\begin{IEEEbiographynophoto}
{Mil\'an Mosonyi} Received his PhD in Physics from the Catholic University of Leuven in 2005. He
joined the Department of Analysis at the Budapest University of Technology and Economics as an assistant professor in 2005, and
he has been an associate professor there since 2012. Currently he is on a research leave at the
Technische Universit\"at M\"unchen -- Institute for Advanced Study.
His main research interests are quantum Shannon theory and quantum statistics.
\end{IEEEbiographynophoto}

\begin{IEEEbiographynophoto}
{Tomohiro Ogawa}
was born in Kanagawa, Japan, in 1969.
He received the B. Eng. and M. Eng. degrees in 1995 and 1997,
respectively, from the University of Tokyo and the Dr. Eng. degree
from the University of Electro-Communications in 2000.

He worked at the University of Tokyo from 2000 to 2005, at the Japan Science
and Technology Agency from 2005 to 2008, and since then he has been with
the University of Electro-Communications.  His research interests
include quantum information theory and information geometry.
\end{IEEEbiographynophoto}


\begin{thebibliography}{99}

\bibitem{Araki1}
H.~Araki,
\kii{Gibbs States of a one dimensional quantum lattice},
\kiii{Commun.~Math.~Phys.} \textbf{14}, 120--157, (1969)

\bibitem{Araki2}
H.~Araki,
\kii{On uniqueness of KMS states of one-dimensional quantum lattice systems},
\kiii{Commun.~Math.~Phys.} \textbf{44}, 1--7, (1975)

\bibitem{Aud} K.M.R.~Audenaert, J.~Calsamiglia, Ll.~Masanes, R.~Munoz-Tapia, A.~Acin, E.~Bagan, F.~Verstraete.:
\kii{Discriminating states: the quantum Chernoff bound};
\kiii{Phys.~Rev.~Lett.} \textbf{98} 160501, (2007)

\bibitem{ANSzV} K.M.R.~Audenaert, M.~Nussbaum, A.~Szko\l a, F.~Verstraete:
\kii{Asymptotic error rates in quantum hypothesis testing};
\kiii{Commun.~Math.~Phys.} {\bf 279}, 251--283 (2008).

\bibitem{Bhatia}
R.~Bhatia: \kii{Matrix Analysis};
\kiii{Springer}, (1997)

\bibitem{Beigi}
Salman Beigi:
\kii{Quantum R\'enyi divergence satisfies data processing inequality};
\kiii{J.~Math.~Phys.}, \textbf{54}, 122202 (2013)

\bibitem{Blahut}
R.E.~Blahut:
\kii{Hypothesis testing and information theory};
\kiii{IEEE Trans.~Inform.~Theory} vol.~20, issue 4, pp.~405--417, (1974)

\bibitem{BR2}
O.~Bratteli, D.W.~Robinson:
\kii{Operator algebras and quantum statistical mechanics II.}; Springer, (1981)

\bibitem{Chen}
Po-Ning Chen:
\kii{Generalization of G\"artner-Ellis Theorem};
\kiii{IEEE Transactions on Information Theory} \textbf{46:7}, pp.~2752--2760, (2000)

\bibitem{Csiszar} I.~Csisz\'ar:
\kii{Generalized cutoff rates and R\'enyi's information measures};
\kiii{IEEE Trans.~Inf.~Theory} \textbf{41}, 26--34, (1995)

\bibitem{DC}
Didier Dacunha-Castelle:
\kii{Formule de Chernoff pour une suite de variables r\'eelles. In: Grandes Deviations et Applications Statistiques};
\kiii{Ast\'erisque} \textbf{68}, pp.~19--24, (1979)

\bibitem{Datta} N.~Datta:
\kii{Min- and Max-Relative Entropies and a New Entanglement Monotone};
\kiii{IEEE Transactions on Information Theory}, vol.~55, no.~6, pp.~2816--2826, (2009).

\bibitem{DZ} A.~Dembo, O.~Zeitouni: {\it Large Deviations Techniques and Applications}\,;
              \kii{Second ed., Springer, Application of Mathematics, Vol.} \textbf{38}, (1998)

\bibitem{DFP}
B.~Dierckx, M.~Fannes, M.~Pogorzelska: \kii{Fermionic Quasi-free States and Maps in Information Theory};
\kiii{J.~Math.~Phys.} \textbf{49}, 032109, (2008)

\bibitem{FNW}
M.~Fannes, B.~Nachtergaele, R.F.~Werner:
\kii{Finitely correlated states on quantum spin chains};
\kiii{Commun.~Math.~Phys.} \textbf{144}, 443--490, (1992)

\bibitem{FL}
Rupert L.~Frank and Elliott H.~Lieb:
\kii{Monotonicity of a relative R\'enyi entropy};
\kiii{J.~Math.~Phys.} 54 , 122201, (2013)

\bibitem{GSz}
U.~Grenander, G.~Szeg\H o:
\kii{Toeplitz Forms and their Applications};
\kiii{University of California, Berkeley}, (1958)

\bibitem{HK}
T.S.~Han and K.~Kobayashi:
\kii{The strong converse theorem for hypothesis testing};
\kiii{IEEE Trans.~Inform.~Theory}, vol. 35, pp. 178--180, (1989)

\bibitem{Han:sc}
T.S.~Han:
\kii{Hypothesis testing with the general source};
\kiii{IEEE Trans.~Inf.~Theory}, vol.~46, pp.~2415--2427, (2000)

\bibitem{Han:book}
T.S.~Han:
\kii{Information-Spectrum Methods in Information Theory};
\kiii{Springer-Verlag}, Berlin, Germany, (2003)

\bibitem{HO}
M.~Hayashi, T.~Ogawa:
\kii{On error exponents in quantum hypothesis testing};
\kiii{IEEE Trans.~Inf.~Theory}, vol.~50, issue 6, pp.~1368--1372, (2004)

\bibitem{H:pinching}
M.~Hayashi,
\kii{Optimal sequence of POVM's in the sense of Stein's lemma
in quantum hypothesis testing};
\kiii{J.~Phys.~A: Math.~Gen.} \textbf{35}, pp.~10759--10773, (2002).

\bibitem{H:text}
M.~Hayashi,
{\em Quantum Information Theory: An Introduction};
Springer, (2006).

\bibitem{Hayashi} M.~Hayashi:
\kii{Error exponent in asymmetric quantum hypothesis testing and its
 application to classical-quantum channel coding};
\kiii{Phys.~Rev.~A} \textbf{76}, 062301, (2007).

\bibitem{Hayashi-ld}
Masahito Hayashi:
\kii{Large deviation analysis for quantum security via smoothing of Renyi entropy of order 2};
\kiii{IEEE Transactions on Information Theory}, Volume 60, Issue 10, pp.~6702--6732, (2014)


\bibitem{HW}
Masahito Hayashi, Shun Watanabe:
\kii{Information Geometry Approach to Parameter Estimation in Markov Chains};
\kiii{arXiv:1401.3814}, (2014)

\bibitem{HT}
Masahito Hayashi, Marco Tomamichel:
\kii{Correlation Detection and an Operational Interpretation of the Renyi Mutual Information};
\kiii{arXiv:1408.6894}

\bibitem{HP}
F.~Hiai, D.~Petz:
\kii{The proper formula for relative entropy and its asymptotics in quantum probability};
\kiii{Comm.~Math.~Phys.} {\bf 143}, 99--114 (1991).

\bibitem{HMO} F.~Hiai, M.~Mosonyi, T.~Ogawa:
\kii{Large deviations and Chernoff bound for certain correlated states on a spin chain};
\kiii{J.~Math.~Phys.} \textbf{48}, (2007)

\bibitem{HMO2} F.~Hiai, M.~Mosonyi, T.~Ogawa:
\kii{Error exponents in hypothesis testing for correlated states on a spin chain};
\kiii{J.~Math.~Phys.} \textbf{49}, 032112, (2008)

\bibitem{Hiai}
F.~Hiai:
\kii{Concavity of certain matrix trace and norm functions};
\kiii{Linear Algebra and Appl.} \textbf{439}, 1568--1589,  (2013)

\bibitem{Hoeffding}
W.~Hoeffding:
\kii{On probabilities of large deviations};
\kiii{Proceedings of Symposium ``the
Fifth Berkeley Symposium on Mathematical Statistics and Probability''}, pp.~203--219,
Berkeley, University of California Press, (1965)

\bibitem{Kato}
T.~Kato:
\kii{Perturbation Theory for Linear Operators};
\kiii{Springer, New York}, (1980)

\bibitem{MHOF}
M.~Mosonyi, F.~Hiai, T.~Ogawa, M.~Fannes:
\kii{Asymptotic distinguishability measures for shift-invariant
quasifree states of fermionic lattice systems};
\kiii{J.~Math.~Phys.} \textbf{49}, 072104, (2008)

\bibitem{M} M.~Mosonyi:
\kii{Hypothesis testing for Gaussian states on bosonic lattices};
\kiii{J.~Math.~Phys.} \textbf{50}, 032104, (2009)

\bibitem{MH}
M.~Mosonyi, F.~Hiai:
\kii{On the quantum R\'enyi relative entropies and related capacity formulas};
\kiii{IEEE Trans.~Inf.~Theory}, \textbf{57}, pp.~2474--2487, (2011)

\bibitem{MO}
Mil\'an Mosonyi, Tomohiro Ogawa:
\kii{Quantum hypothesis testing and the operational interpretation of the quantum R\'enyi relative entropies};
\kiii{Communications in Mathematical Physics}, Volume 334, Issue 3, pp.~1617--1648, (2015)

\bibitem{Mosonyi}
Mil\'an Mosonyi,
\kii{Coding theorems for compound problems via quantum R\'enyi divergences};
\kiii{arXiv:1310.7525, (2013); to appear in IEEE Transactions on Information Theory}

\bibitem{Renyi_new}
Martin M\"uller-Lennert, Fr\'ed\'eric Dupuis, Oleg Szehr, Serge Fehr, Marco Tomamichel:
\kii{On quantum Renyi entropies: a new definition and some properties};
\kiii{J.~Math.~Phys.} \textbf{54}, 122203, (2013)

\bibitem{Nagaoka2} H.~Nagaoka:
\kii{Strong converse theorems in quantum information theory};
\kiii{in the book ``Asymptotic Theory of Quantum Statistical Inference'' edited by M.~Hayashi},
World Scientific, (2005)

\bibitem{Nagaoka} H.~Nagaoka:
\kii{The converse part of the theorem for quantum Hoeffding bound};
\kiii{quant-ph/0611289}

\bibitem{NH} H.~Nagaoka, M.~Hayashi:
\kii{An information-spectrum approach to classical and quantum hypothesis testing for simple hypotheses};
\kiii{IEEE Trans.~Inform.~Theory} \textbf{53}, 534--549, (2007)

\bibitem{NKiid}
Kenji Nakagawa, Fumio Kanaya:
\kii{On the converse theorem in statistical hypothesis testing};
\kiii{IEEE Transactions on Information Theory}, vol.~39, no.~2, pp.~623--628, (1993)

\bibitem{NKMarkov}
Kenji Nakagawa, Fumio Kanaya:
\kii{On the converse theorem in statistical hypothesis testing for Markov chains};
\kiii{IEEE Transactions on Information Theory}, vol.~39, no.~2, pp.~629--633, (1993)

\bibitem{Natarajan}
S.~Natarajan:
\kii{Large deviations, hypothesis testing, and source coding for finite Markov chains};
\kiii{IEEE Transactions on Information Theory} \textbf{31}, pp.~360--365, (1985)

\bibitem{NSz} M.~Nussbaum,  A.~Szko{\l}a:
\kii{A lower bound of Chernoff type for symmetric quantum hypothesis testing};
\kiii{Ann.~Statist.} \textbf{37}, 1040--1057, (2009)

\bibitem{ON}
T.~Ogawa, H.~Nagaoka:
\kii{Strong converse and Stein's lemma in quantum hypothesis testing};
\kiii{IEEE Trans. Inform. Theory} {\bf 47}, 2428--2433 (2000).

\bibitem{OP}
M.~Ohya, D.~Petz:
\kii{Quantum Entropy and its Use};
\kiii{Springer}, (1993)

\bibitem{Petz} D.~Petz:
\kii{Quasi-entropies for finite quantum systems};
\kiii{Rep.~Math.~Phys.} \textbf{23}, 57--65, (1986)

\bibitem{Renyi}
A.~R\'enyi:
\kii{On measures of entropy and information};
\kiii{Proc.~4th Berkeley Sympos.~Math.~Statist.~and Prob.}, Vol.~I, pp.~547--561, Univ. California Press, Berkeley, California,  (1961)

\bibitem{RennerPhD}
R.~Renner:
{\em Security of Quantum Key Distribution}, PhD dissertation, Swiss Federal Institute of Technology Zurich, Diss.~ETH No.~16242, (2005).

\bibitem{TH}
Marco Tomamichel, Masahito Hayashi:
\kii{A Hierarchy of information quantities for finite block length analysis of quantum tasks};
\kiii{IEEE Transactions on Information Theory}, Volume:	59	, Issue: 11, pp.~7693--7710, (2013)

\bibitem{Umegaki}
H.~Umegaki: \kii{Conditional expectation in an operator algebra};
\kiii{Kodai Math.~Sem.~Rep.} \textbf{14}, 59--85, (1962)

\bibitem{Vasek}
Karol Va\v sek:
\kii{On the error exponent for ergodic Markov source};
\kiii{Kybernetika}, vol.~16, no.~4, pp.~318--329, (1980)

\bibitem{WH}
Shun Watanabe, Masahito Hayashi:
\kii{Finite-length Analysis on Tail probability for Markov Chain and Application to Simple Hypothesis Testing};
\kiii{arXiv:1401.3801}, (2014)

\bibitem{Wehrl}
Alfred Wehrl:
\kii{General properties of entropy};
\kiii{Rev.~Mod.~Phys.} \textbf{50}, pp.~221--260, (1978)

\bibitem{WWY}
Mark M.~Wilde, Andreas Winter, Dong Yang:
\kii{Strong converse for the classical capacity of entanglement-breaking and Hadamard channels};
\kiii{Communications in Mathematical Physics}, \textbf{331}, pp.~593--622, (2014)


\end{thebibliography}
\end{document}